\documentclass[]{article}

\pdfoutput=1

\usepackage{amsthm}
\usepackage{graphics}
\usepackage[utf8]{inputenc}
\usepackage[backend=biber,  doi=false,isbn=false,
style=numeric]{biblatex}
\usepackage{abstract}

\usepackage{amsfonts}
\usepackage{amsmath}

\usepackage[font={bf, sf, footnotesize,singlespacing},format=hang, skip=0pt,
margin = 0pt,justification=raggedright,
indention=0em,position=bottom]{subcaption}
\usepackage{tikz}
\usepackage[english,capitalize,noabbrev]{cleveref}
\usepackage{xspace}
\usepackage{booktabs} 
\usetikzlibrary{calc}
\usetikzlibrary{arrows.meta}
\usetikzlibrary{spy}
\usepackage{dsfont}

\usepackage{array}

\usepackage{url}

\newif\ifblackandwhitecycle
\gdef\patternnumber{0}

\definecolor{TUMdarkblue}{HTML}{002143}
\definecolor{TUMboxbg}{HTML}{DBE2E9}
\definecolor{TUMhighlight}{HTML}{11639D}

\DeclareMathOperator{\relint}{relint}

\pgfkeys{/tikz/.cd,
    zoombox paths/.style={
        draw=orange,
        very thick
    },
    black and white/.is choice,
    black and white/.default=static,
    black and white/static/.style={ 
        draw=white,   
        zoombox paths/.append style={
            draw=white,
            postaction={
                draw=black,
                loosely dashed
            }
        }
    },
    black and white/static/.code={
        \gdef\patternnumber{1}
    },
    black and white/cycle/.code={
        \blackandwhitecycletrue
        \gdef\patternnumber{1}
    },
    black and white pattern/.is choice,
    black and white pattern/0/.style={},
    black and white pattern/1/.style={    
            draw=white,
            postaction={
                draw=black,
                dash pattern=on 2pt off 2pt
            }
    },
    black and white pattern/2/.style={    
            draw=white,
            postaction={
                draw=black,
                dash pattern=on 4pt off 4pt
            }
    },
    black and white pattern/3/.style={    
            draw=white,
            postaction={
                draw=black,
                dash pattern=on 4pt off 4pt on 1pt off 4pt
            }
    },
    black and white pattern/4/.style={    
            draw=white,
            postaction={
                draw=black,
                dash pattern=on 4pt off 2pt on 2 pt off 2pt on 2 pt off 2pt
            }
    },
    zoomboxarray inner gap/.initial=5pt,
    zoomboxarray columns/.initial=2,
    zoomboxarray rows/.initial=2,
    subfigurename/.initial={},
    figurename/.initial={zoombox},
    zoomboxarray/.style={
        execute at begin picture={
            \begin{scope}[
                spy using outlines={%
                    zoombox paths,
                    width=\imagewidth / \pgfkeysvalueof{/tikz/zoomboxarray columns} - (\pgfkeysvalueof{/tikz/zoomboxarray columns} - 1) / \pgfkeysvalueof{/tikz/zoomboxarray columns} * \pgfkeysvalueof{/tikz/zoomboxarray inner gap} -\pgflinewidth,
                    height=\imageheight / \pgfkeysvalueof{/tikz/zoomboxarray rows} - (\pgfkeysvalueof{/tikz/zoomboxarray rows} - 1) / \pgfkeysvalueof{/tikz/zoomboxarray rows} * \pgfkeysvalueof{/tikz/zoomboxarray inner gap}-\pgflinewidth,
                    magnification=3,
                    every spy on node/.style={
                        zoombox paths
                    },
                    every spy in node/.style={
                        zoombox paths
                    }
                }
            ]
        },
        execute at end picture={
            \end{scope}
     \gdef\patternnumber{0}
        },
        spymargin/.initial=0.5em,
        zoomboxes xshift/.initial=1,
        zoomboxes right/.code=\pgfkeys{/tikz/zoomboxes xshift=1},
        zoomboxes left/.code=\pgfkeys{/tikz/zoomboxes xshift=-1},
        zoomboxes yshift/.initial=0,
        zoomboxes above/.code={
            \pgfkeys{/tikz/zoomboxes yshift=1},
            \pgfkeys{/tikz/zoomboxes xshift=0}
        },
        zoomboxes below/.code={
            \pgfkeys{/tikz/zoomboxes yshift=-1},
            \pgfkeys{/tikz/zoomboxes xshift=0}
        },
        caption margin/.initial=4ex,
    },
    adjust caption spacing/.code={},
    image container/.style={
        inner sep=0pt,
        at=(image.north),
        anchor=north,
        adjust caption spacing
    },
    zoomboxes container/.style={
        inner sep=0pt,
        at=(image.north),
        anchor=north,
        name=zoomboxes container,
        xshift=\pgfkeysvalueof{/tikz/zoomboxes xshift}*(\imagewidth+\pgfkeysvalueof{/tikz/spymargin}),
        yshift=\pgfkeysvalueof{/tikz/zoomboxes yshift}*(\imageheight+\pgfkeysvalueof{/tikz/spymargin}+\pgfkeysvalueof{/tikz/caption margin}),
        adjust caption spacing
    },
    calculate dimensions/.code={
        \pgfpointdiff{\pgfpointanchor{image}{south west} }{\pgfpointanchor{image}{north east} }
        \pgfgetlastxy{\imagewidth}{\imageheight}
        \global\let\imagewidth=\imagewidth
        \global\let\imageheight=\imageheight
        \gdef\columncount{1}
        \gdef\rowcount{1}
        
    },
    image node/.style={
        inner sep=0pt,
        name=image,
        anchor=south west,
        append after command={
            [calculate dimensions]
            node [image container,subfigurename=\pgfkeysvalueof{/tikz/figurename}-image] {\phantomimage}
            node [zoomboxes container,subfigurename=\pgfkeysvalueof{/tikz/figurename}-zoom] {\phantomimage}
        }
    },
    color code/.style={
        zoombox paths/.append style={draw=#1}
    },
    connect zoomboxes/.style={
    spy connection path={\draw[draw=none,zoombox paths] (tikzspyonnode) -- (tikzspyinnode);}
    },
    help grid code/.code={
        \begin{scope}[
                x={(image.south east)},
                y={(image.north west)},
                font=\footnotesize,
                help lines,
                overlay
            ]
            \foreach \x in {0,1,...,9} { 
                \draw(\x/10,0) -- (\x/10,1);
                \node [anchor=north] at (\x/10,0) {0.\x};
            }
            \foreach \y in {0,1,...,9} {
                \draw(0,\y/10) -- (1,\y/10);                        \node [anchor=east] at (0,\y/10) {0.\y};
            }
        \end{scope}    
    },
    help grid/.style={
        append after command={
            [help grid code]
        }
    },
}

\newcommand\phantomimage{%
    \phantom{%
        \rule{\imagewidth}{\imageheight}%
    }%
}
\newcommand\zoombox[2][]{
    \begin{scope}[zoombox paths]
        \pgfmathsetmacro\xpos{
            (\columncount-1)*(\imagewidth / \pgfkeysvalueof{/tikz/zoomboxarray columns} + \pgfkeysvalueof{/tikz/zoomboxarray inner gap} / \pgfkeysvalueof{/tikz/zoomboxarray columns} ) + \pgflinewidth
        }
        \pgfmathsetmacro\ypos{
            (\rowcount-1)*( \imageheight / \pgfkeysvalueof{/tikz/zoomboxarray rows} + \pgfkeysvalueof{/tikz/zoomboxarray inner gap} / \pgfkeysvalueof{/tikz/zoomboxarray rows} ) + 0.5*\pgflinewidth
        }
        \edef\dospy{\noexpand\spy [
            #1,
            zoombox paths/.append style={
                black and white pattern=\patternnumber
            },
            every spy on node/.append style={#1},
            x=\imagewidth,
            y=\imageheight
        ] on (#2) in node [anchor=north west] at ($(zoomboxes container.north west)+(\xpos pt,-\ypos pt)$);}
        \dospy
        \pgfmathtruncatemacro\pgfmathresult{ifthenelse(\columncount==\pgfkeysvalueof{/tikz/zoomboxarray columns},\rowcount+1,\rowcount)}
        \global\let\rowcount=\pgfmathresult
        \pgfmathtruncatemacro\pgfmathresult{ifthenelse(\columncount==\pgfkeysvalueof{/tikz/zoomboxarray columns},1,\columncount+1)}
        \global\let\columncount=\pgfmathresult
        \ifblackandwhitecycle
            \pgfmathtruncatemacro{\newpatternnumber}{\patternnumber+1}
            \global\edef\patternnumber{\newpatternnumber}
        \fi
    \end{scope}
}

\usepackage{todonotes}
\usepackage[]{appendix}

\usepackage{rotating}

\newcommand{\figureDevGermanyElection}{
	\includegraphics[width=123.96469pt, height=162.7443pt]{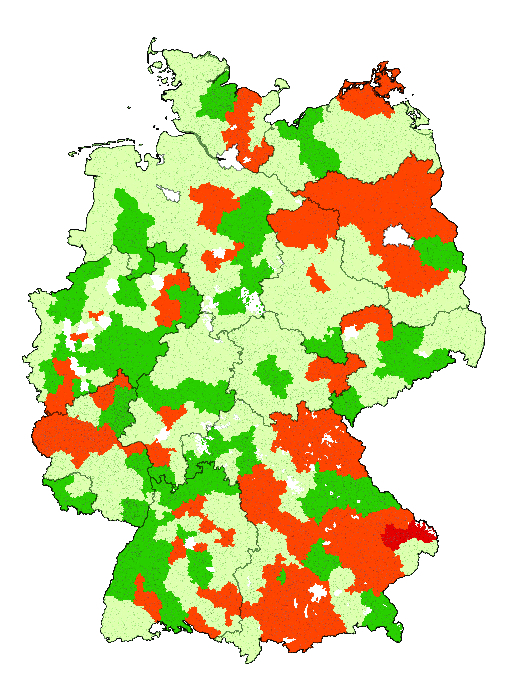}
}

\newcommand{\figureDevGermanySP}{
	\includegraphics[width=123.96469pt, height=162.7443pt]{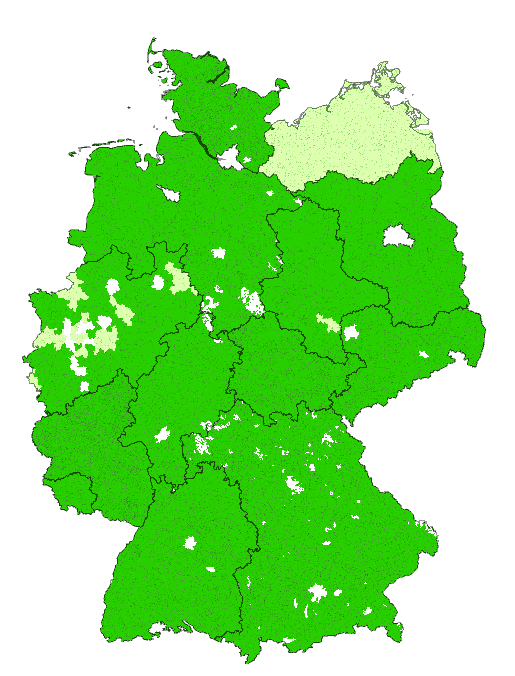}
}

\newcommand{\figureDevLegend}{
	\includegraphics{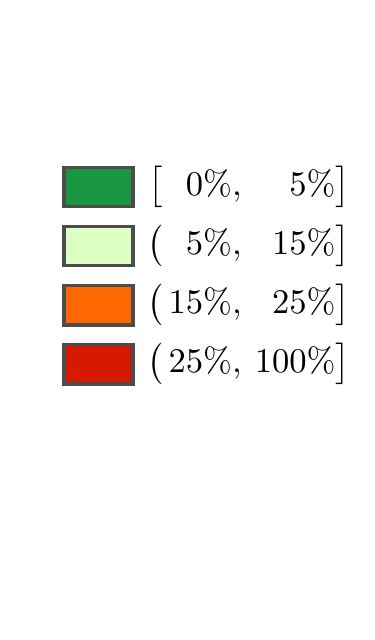}
}

\newcommand{\figureExampleOnlyAffineTrans}{
	\includegraphics{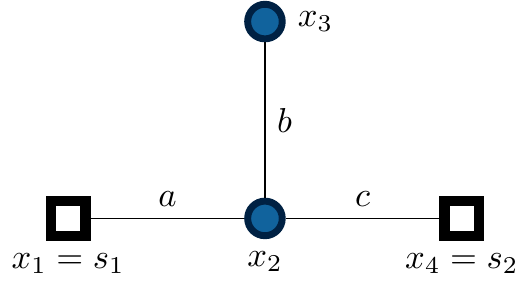}
}

\newcommand{\figureZoomedClusteringFracts}{
	\includegraphics{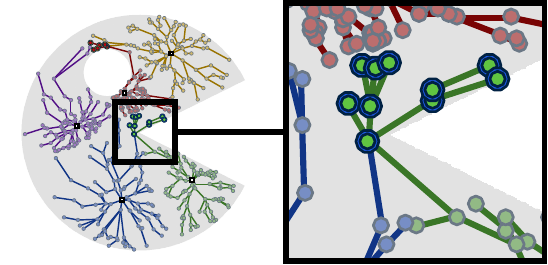}
}

\newcommand{\figureZoomedClusterinSquaredConnected}{
	\includegraphics{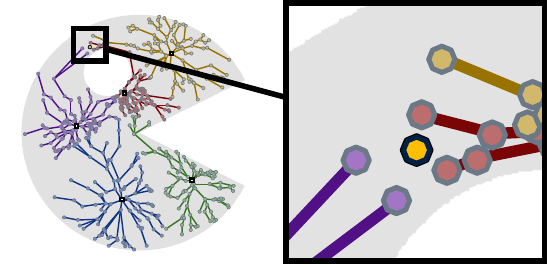}
}

\newcommand{\figureBayernOriginalDistricts}{
	\includegraphics[width=172.26447pt, height=177.14032pt]{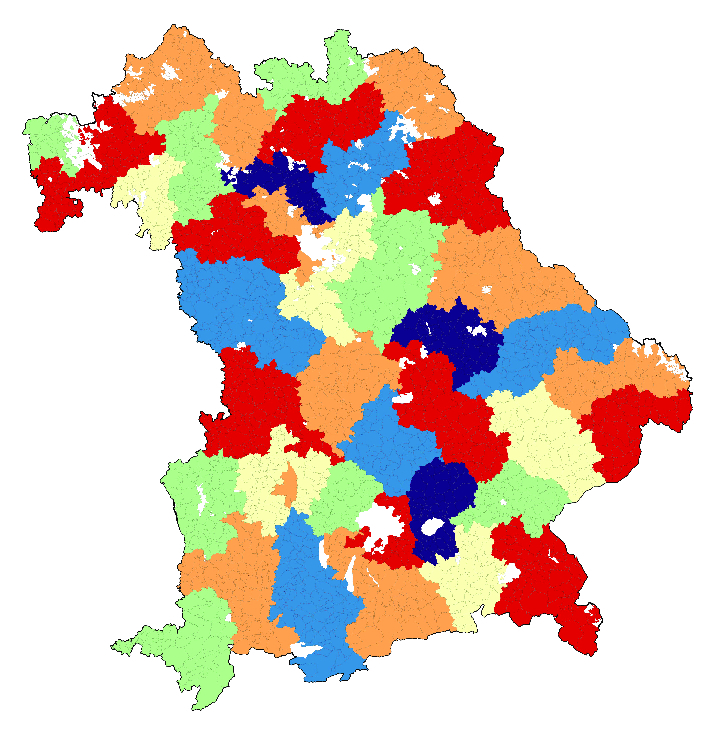}
}

\newcommand{\figureBayernPowerDiagram}{
	\includegraphics[width=172.26447pt, height=177.14032pt]{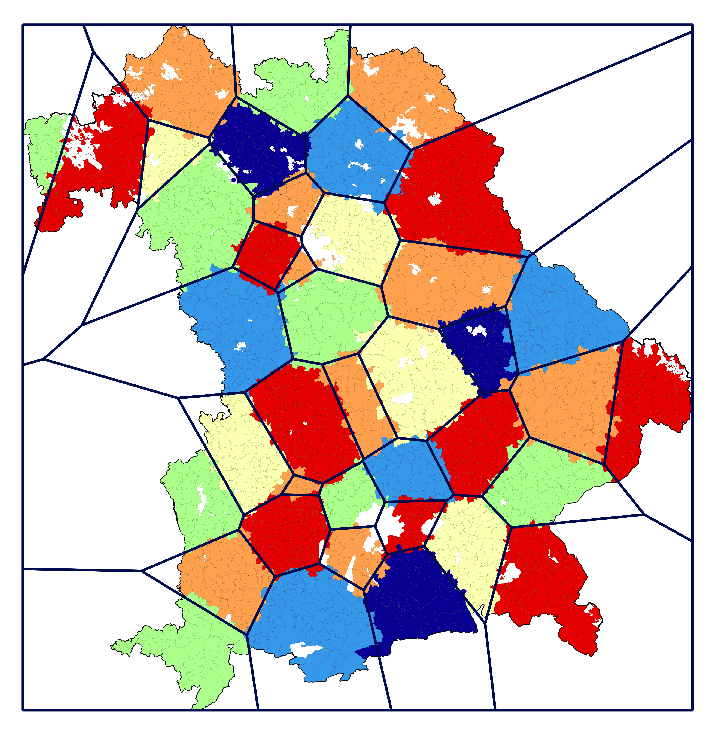}
}

\newcommand{\figureNiedersachsenOriginalDistricts}{
	\includegraphics[width=161.91489pt, height=146.81049pt]{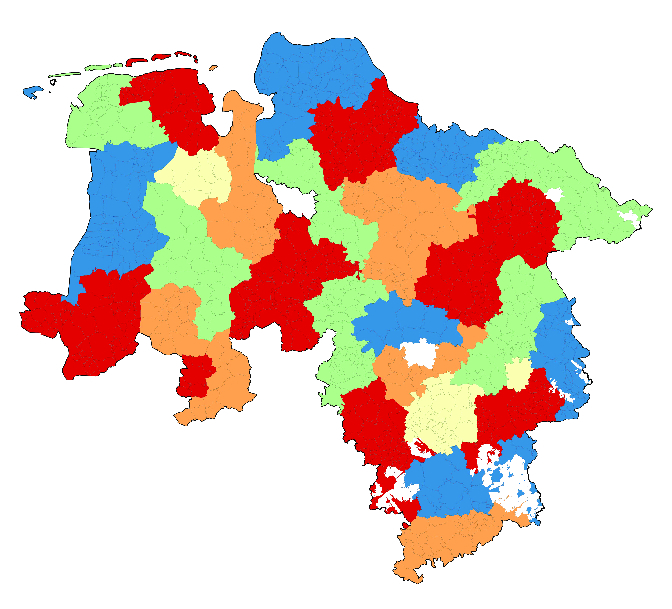}
}

\newcommand{\figureNiedersachsenAPD}{
	\includegraphics[width=161.91489pt, height=146.81049pt]{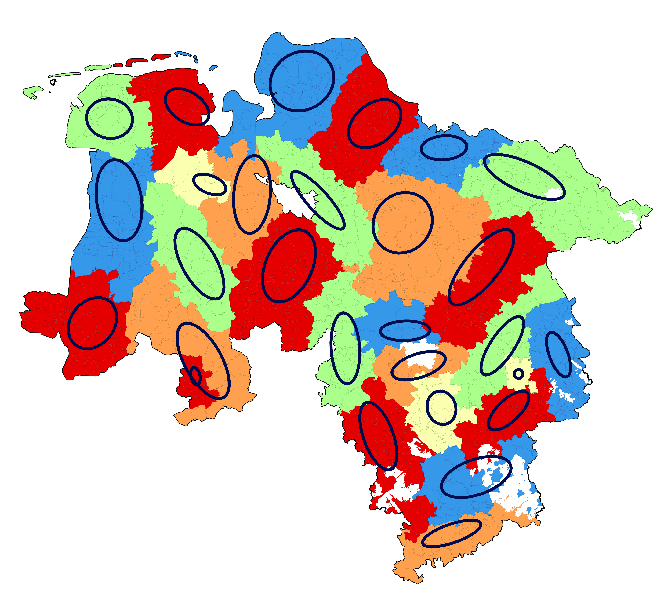}
}

\newcommand{\figureNRWOriginalDistricts}{
	\includegraphics[width=154.15198pt, height=155.17471pt]{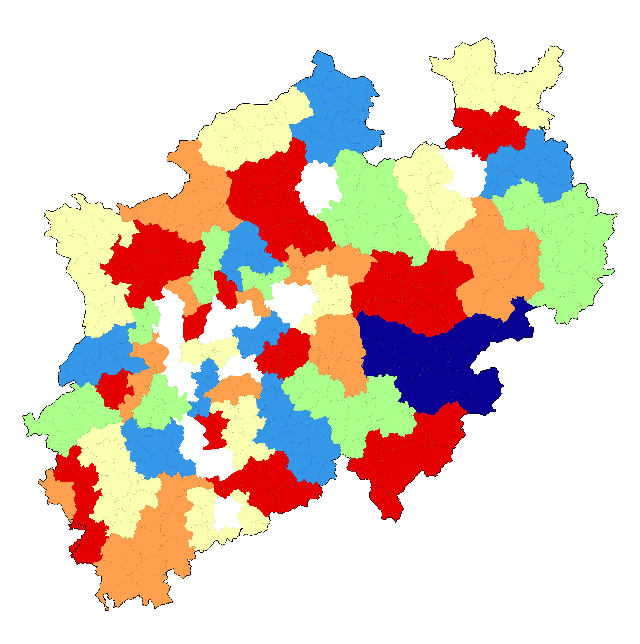}
}

\newcommand{\figureNRWSP}{
	\includegraphics[width=154.15198pt, height=155.17471pt]{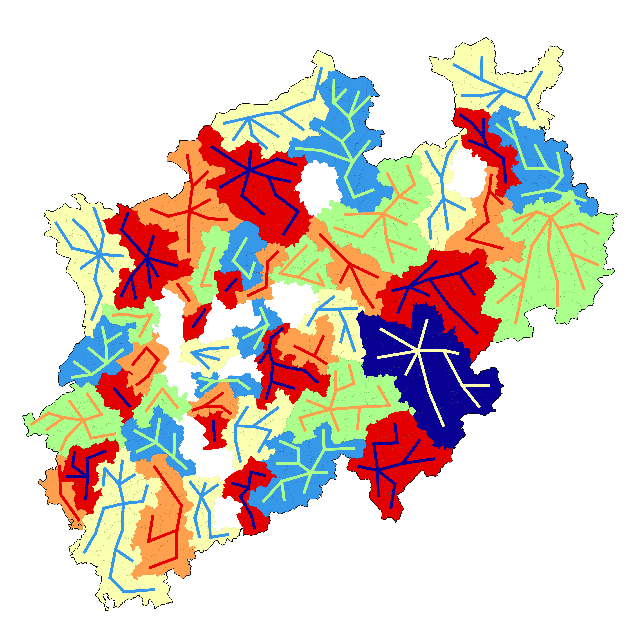}
}

\newcommand{\figureBundDevPD}{
	\includegraphics[width=113.61511pt, height=148.96805pt]{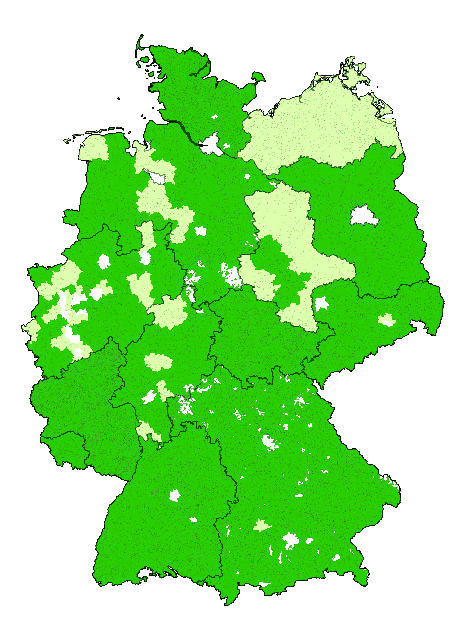}
}

\newcommand{\figureBundDevAPD}{
	\includegraphics[width=113.61511pt, height=148.96805pt]{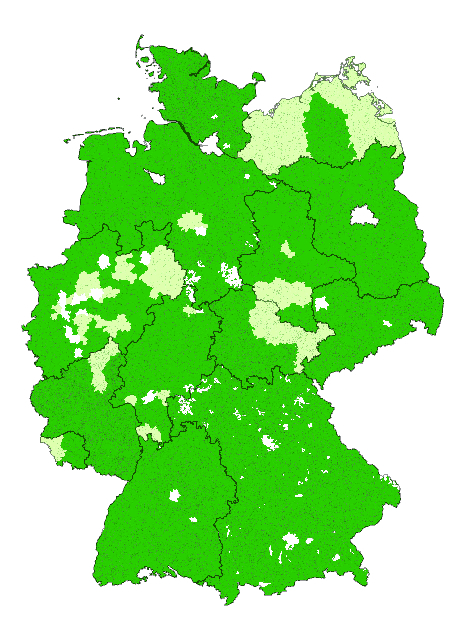}
}

\newcommand{\figureBundDevSP}{
	\includegraphics[width=113.61511pt, height=148.96805pt]{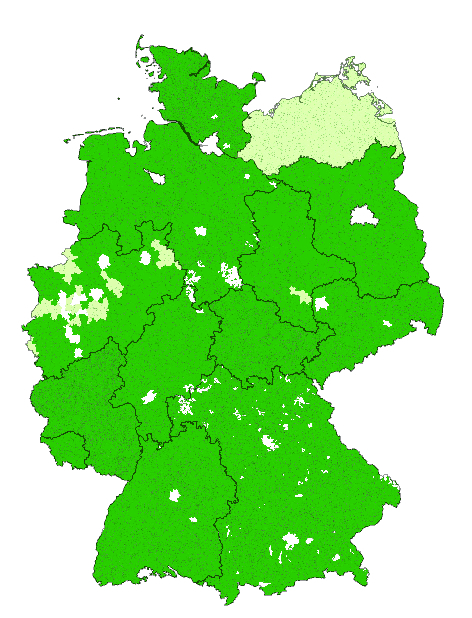}
}

\newcommand{\figureHessenOriginalDistricts}{
	\includegraphics[width=163.98404pt, height=201.69873pt]{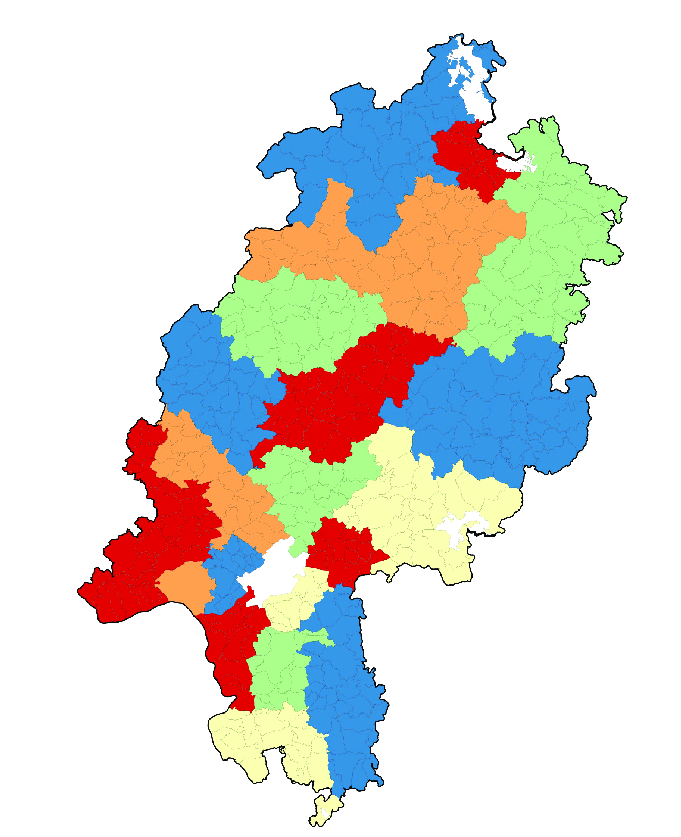}
}

\newcommand{\figureHessenPD}{
	\includegraphics[width=163.98404pt, height=201.69873pt]{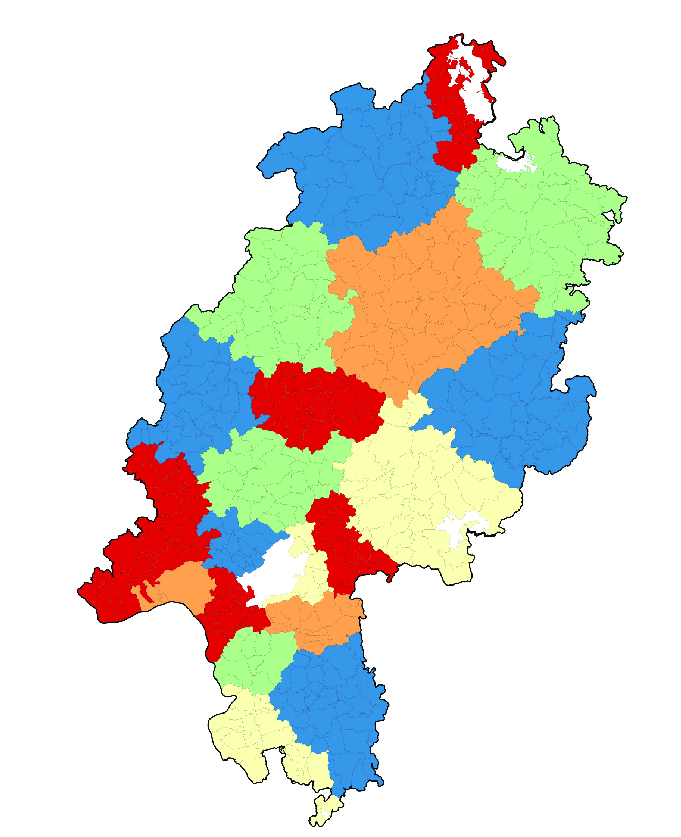}
}

\newcommand{\figureHessenAPD}{
	\includegraphics[width=163.98404pt, height=201.69873pt]{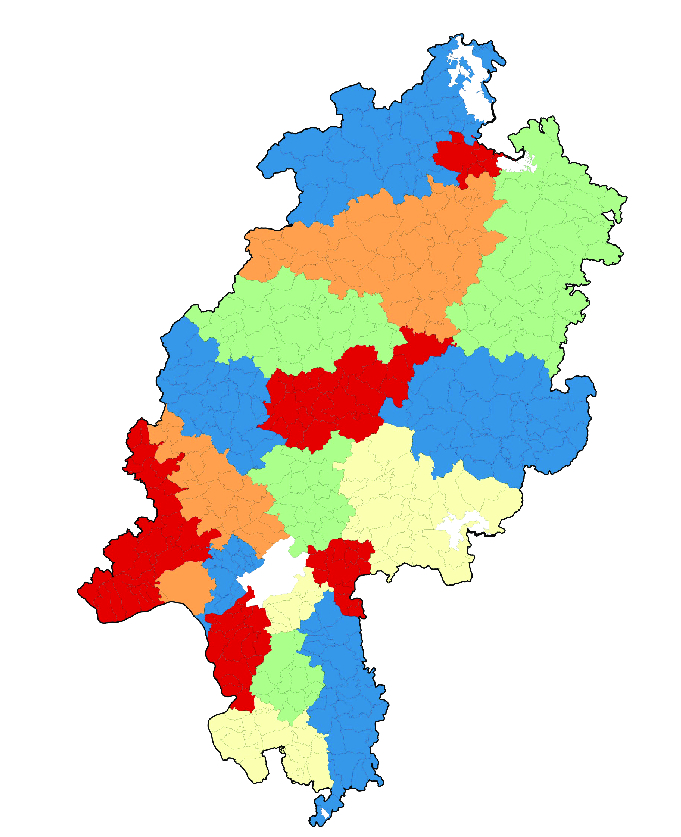}
}

\newcommand{\figureHessenSP}{
	\includegraphics[width=163.98404pt, height=201.69873pt]{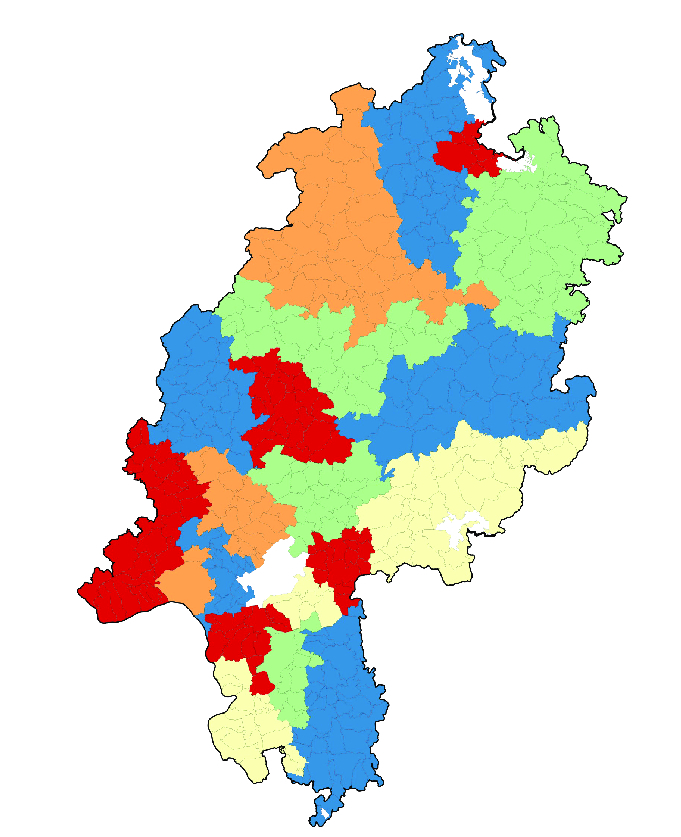}
}

\newcommand{\figureToyExPDDiagram}{
	\includegraphics[width=\textwidth]{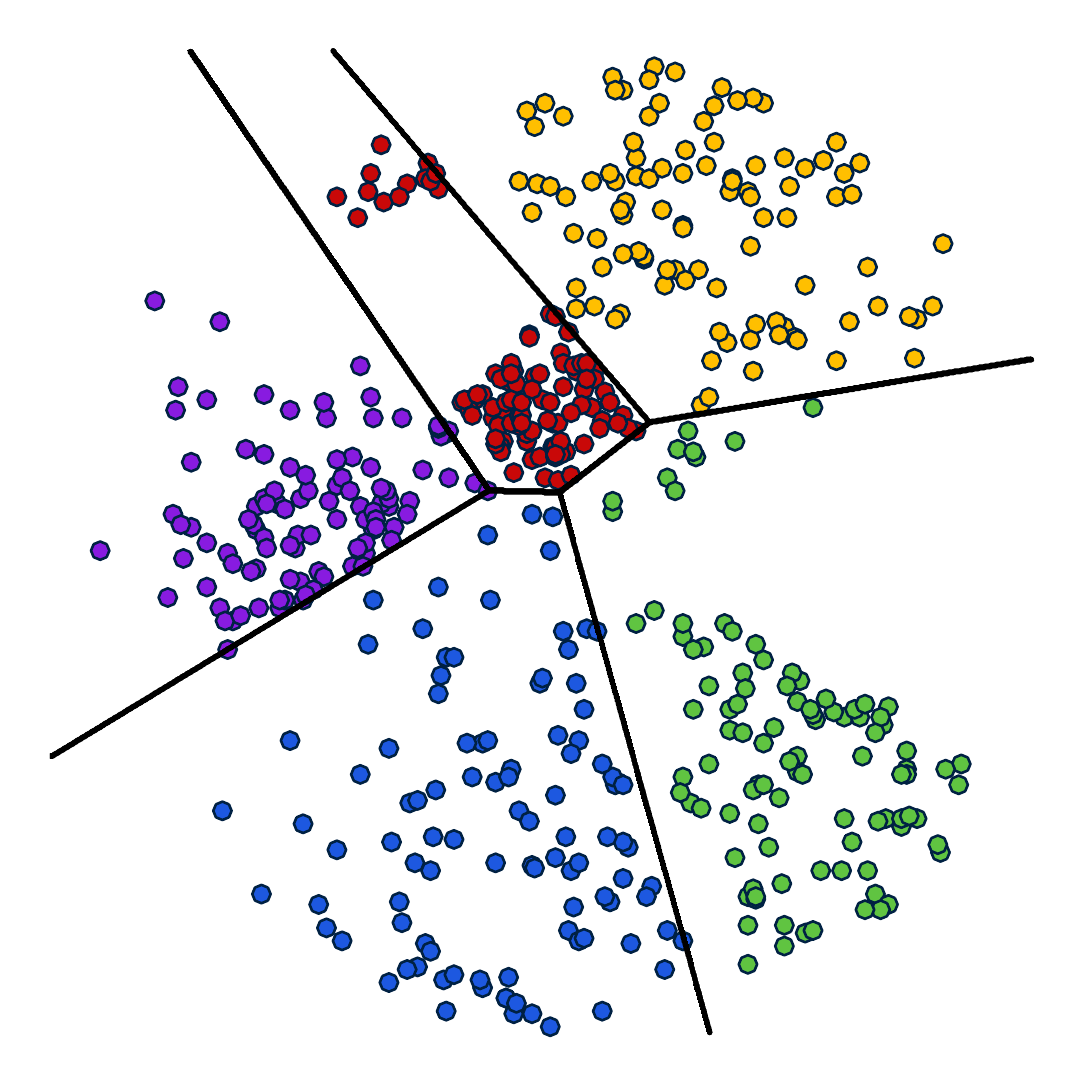}
}

\newcommand{\figureToyExAPDDiagram}{
 \includegraphics[width=\textwidth]{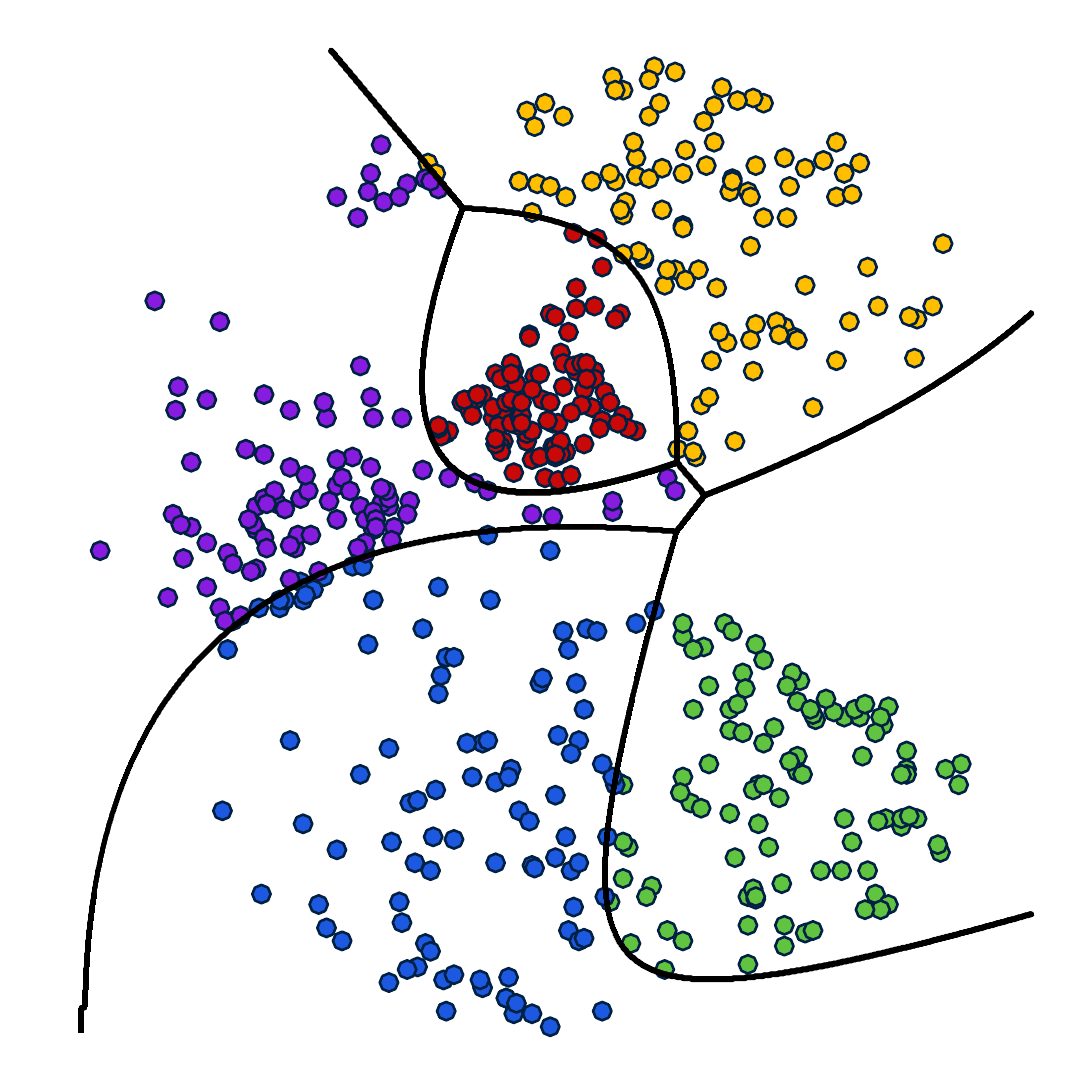}
}

\newcommand{\figureToyExSPDiagram}{
\includegraphics[width=\textwidth]{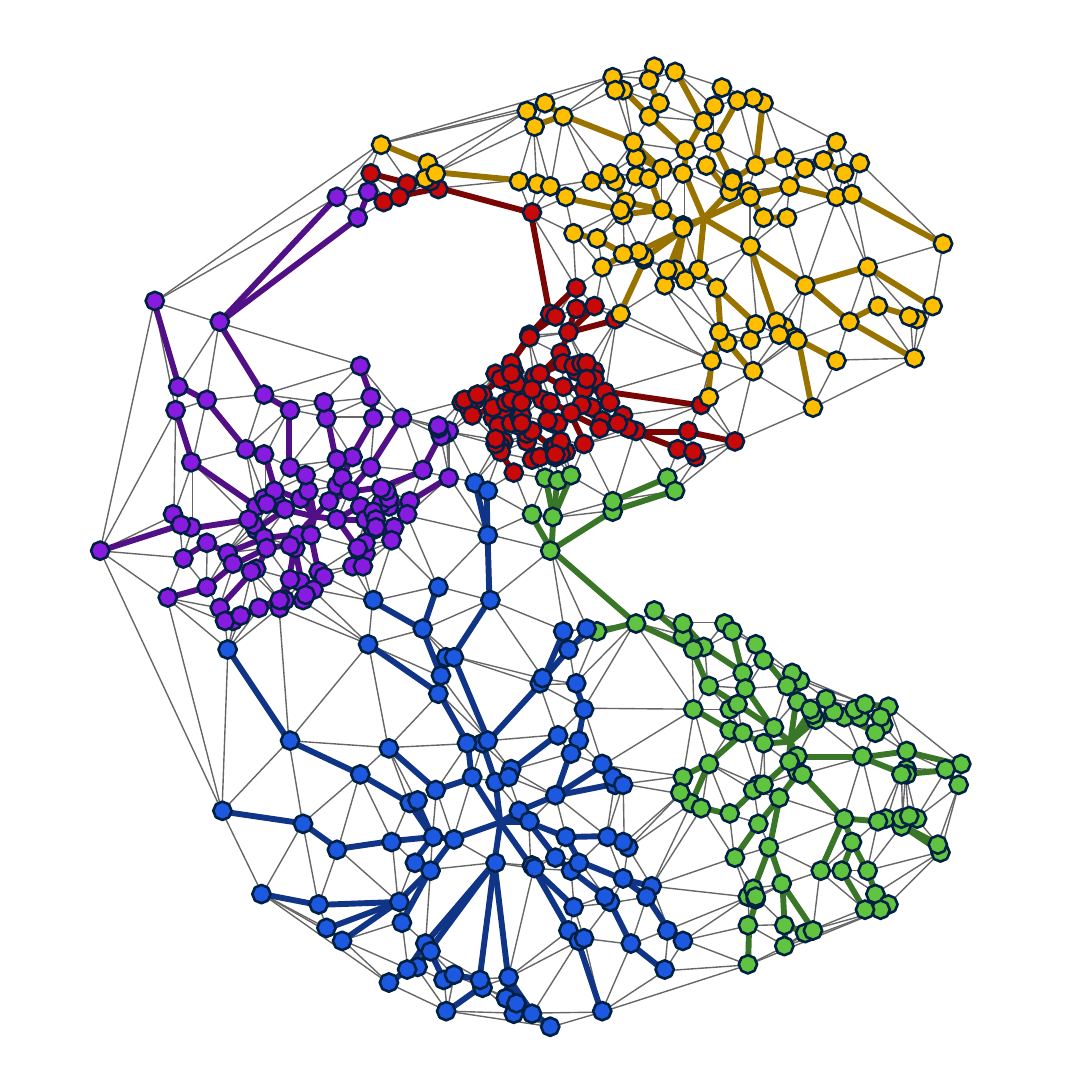}
}

\newcommand{\figureToyEx}[1][0.43\textwidth]{
	\includegraphics[width = #1]{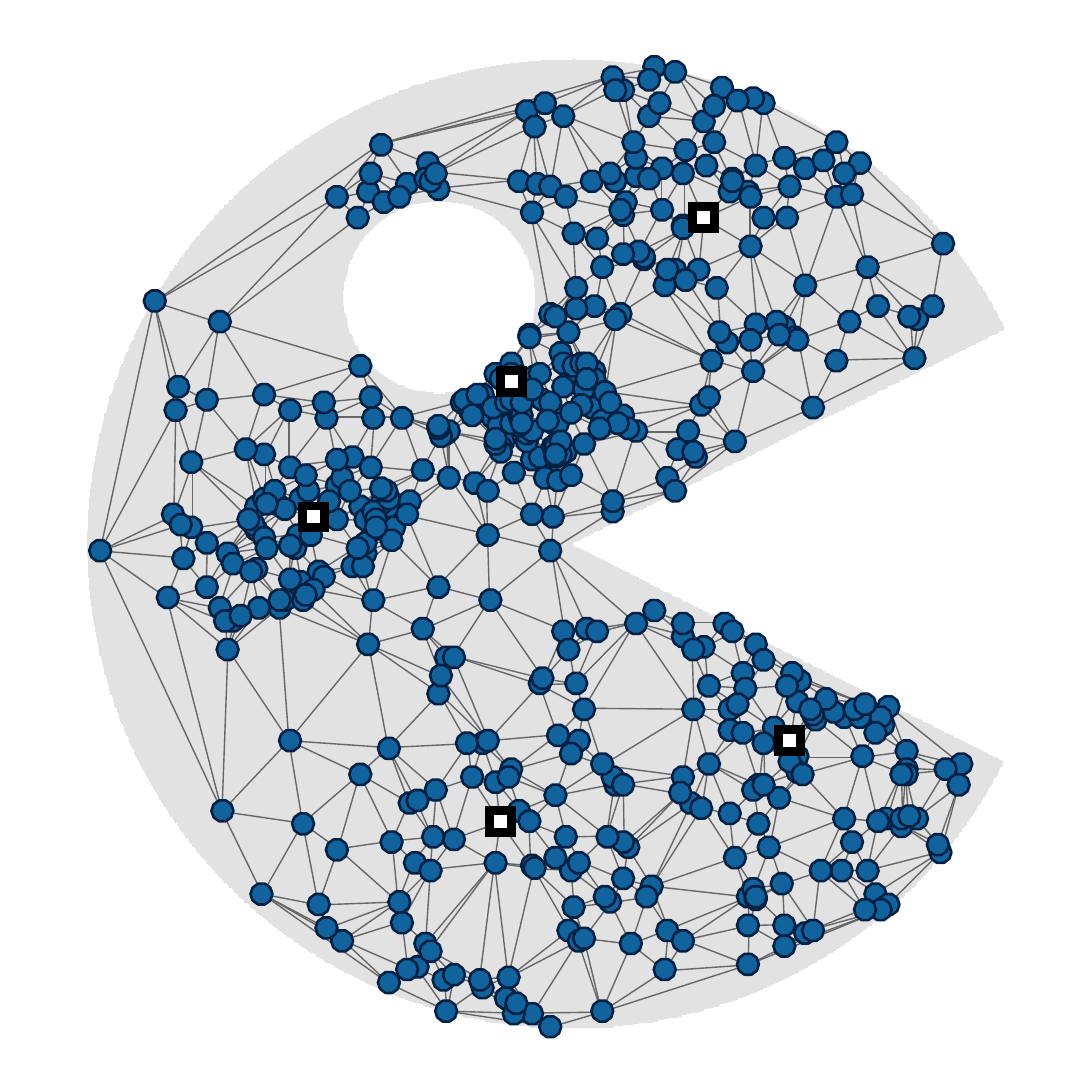}
}

\newcommand{\figureToyExAWVD}[1][0.43\textwidth]{
	\includegraphics[width = #1]{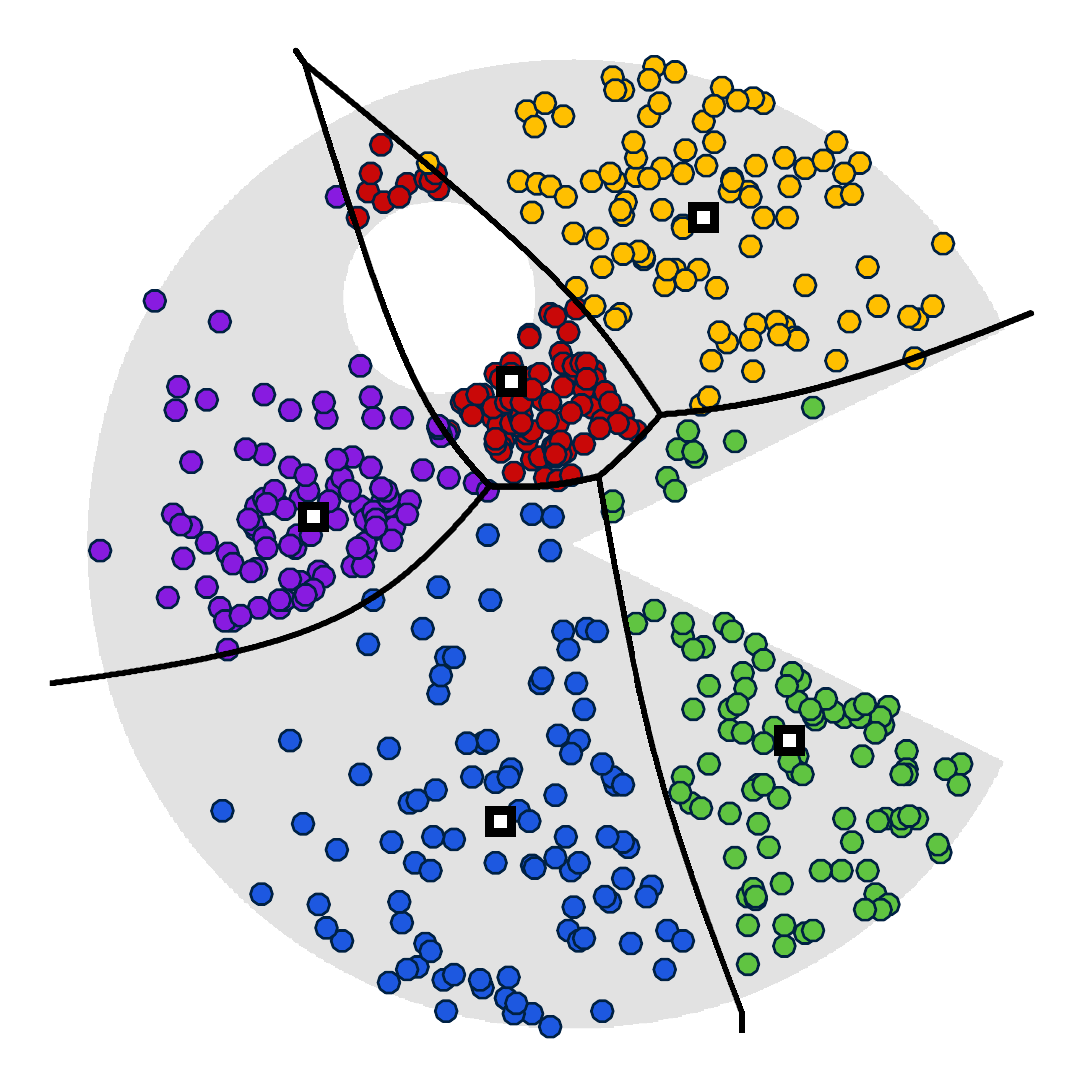}
}

\newcommand{\figureToyExPD}[1][0.43\textwidth]{
	\includegraphics[width = #1]{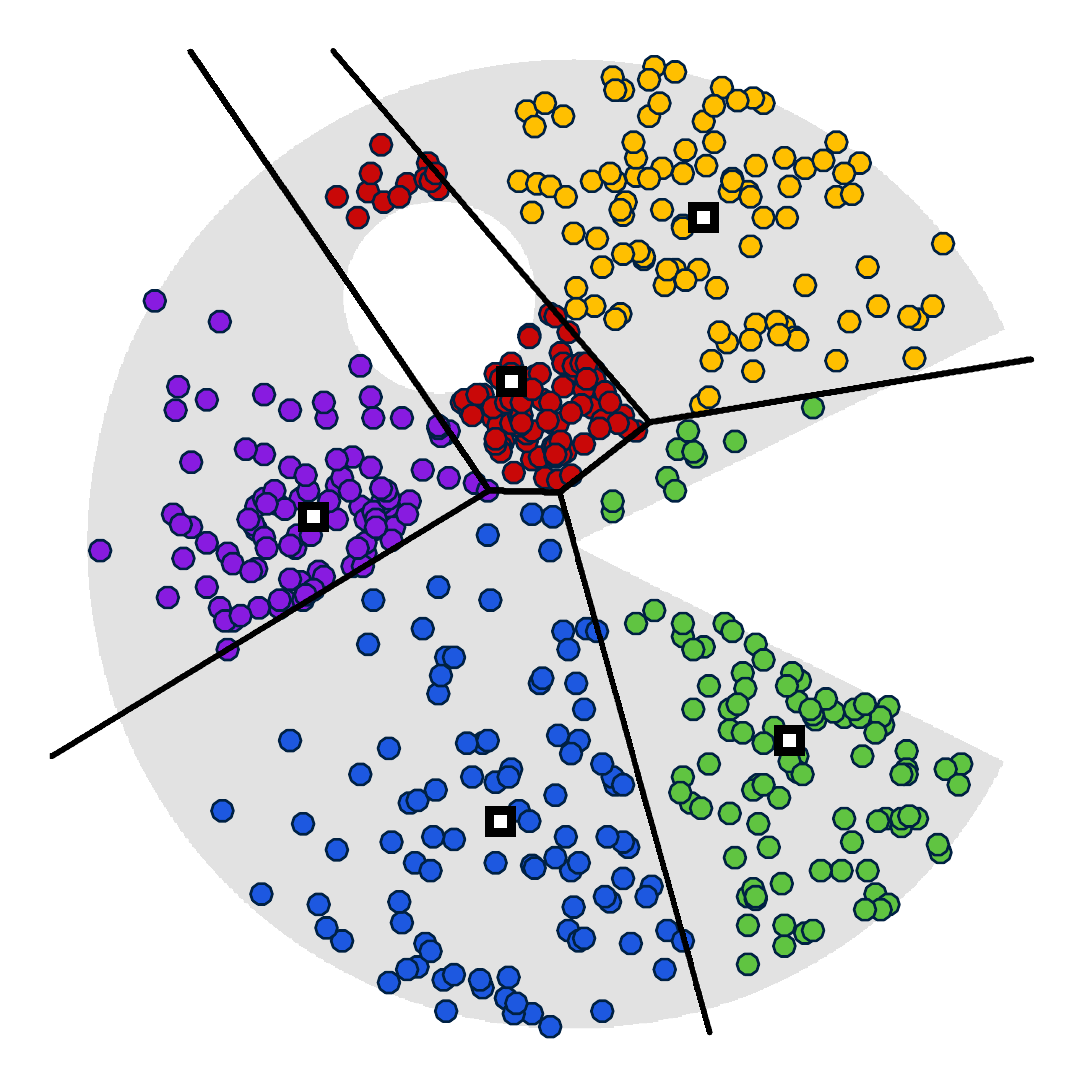}
}

\newcommand{\figureToyExAVD}[1][0.43\textwidth]{
	\includegraphics[width = #1]{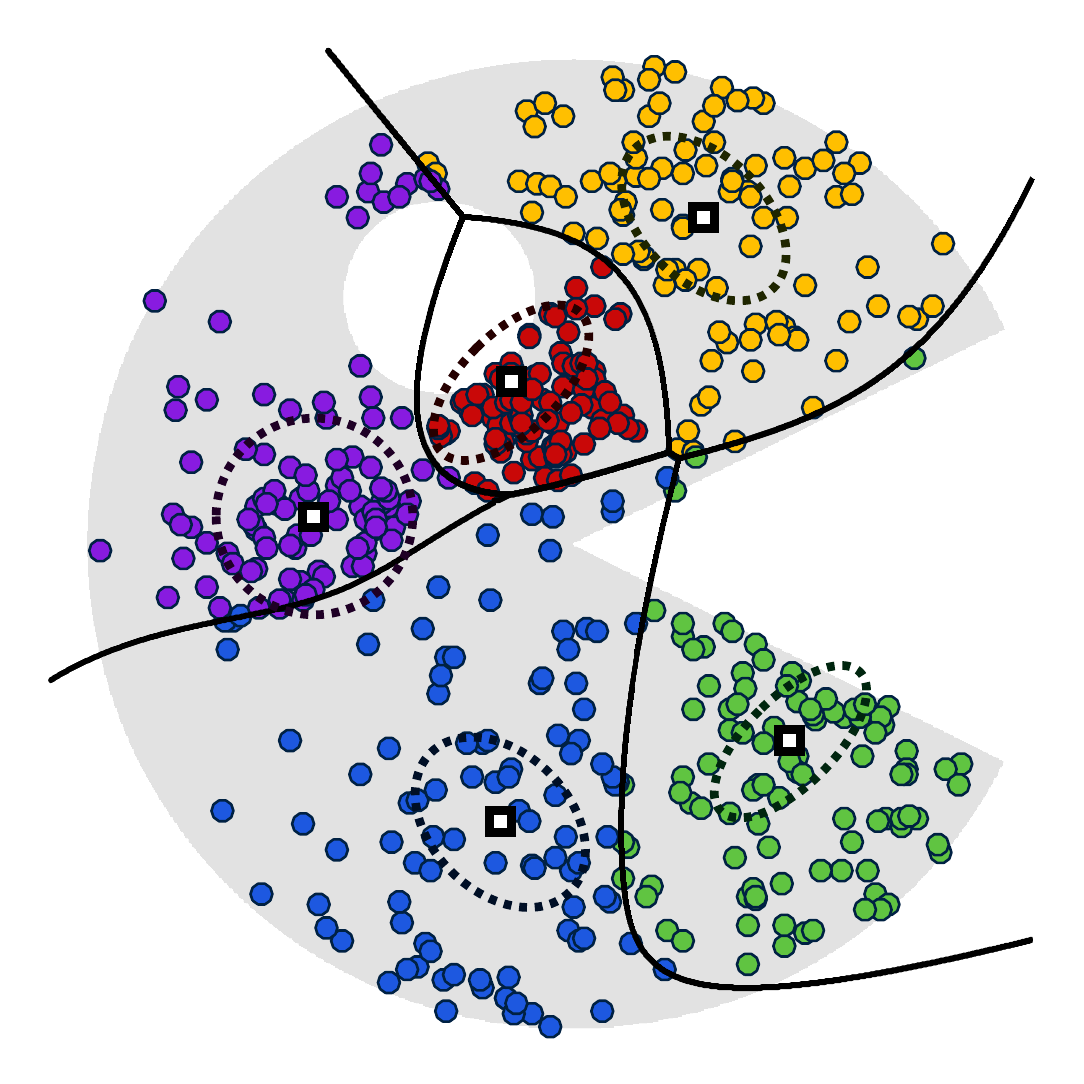}
}

\newcommand{\figureToyExAPD}[1][0.43\textwidth]{
	\includegraphics[width =#1]{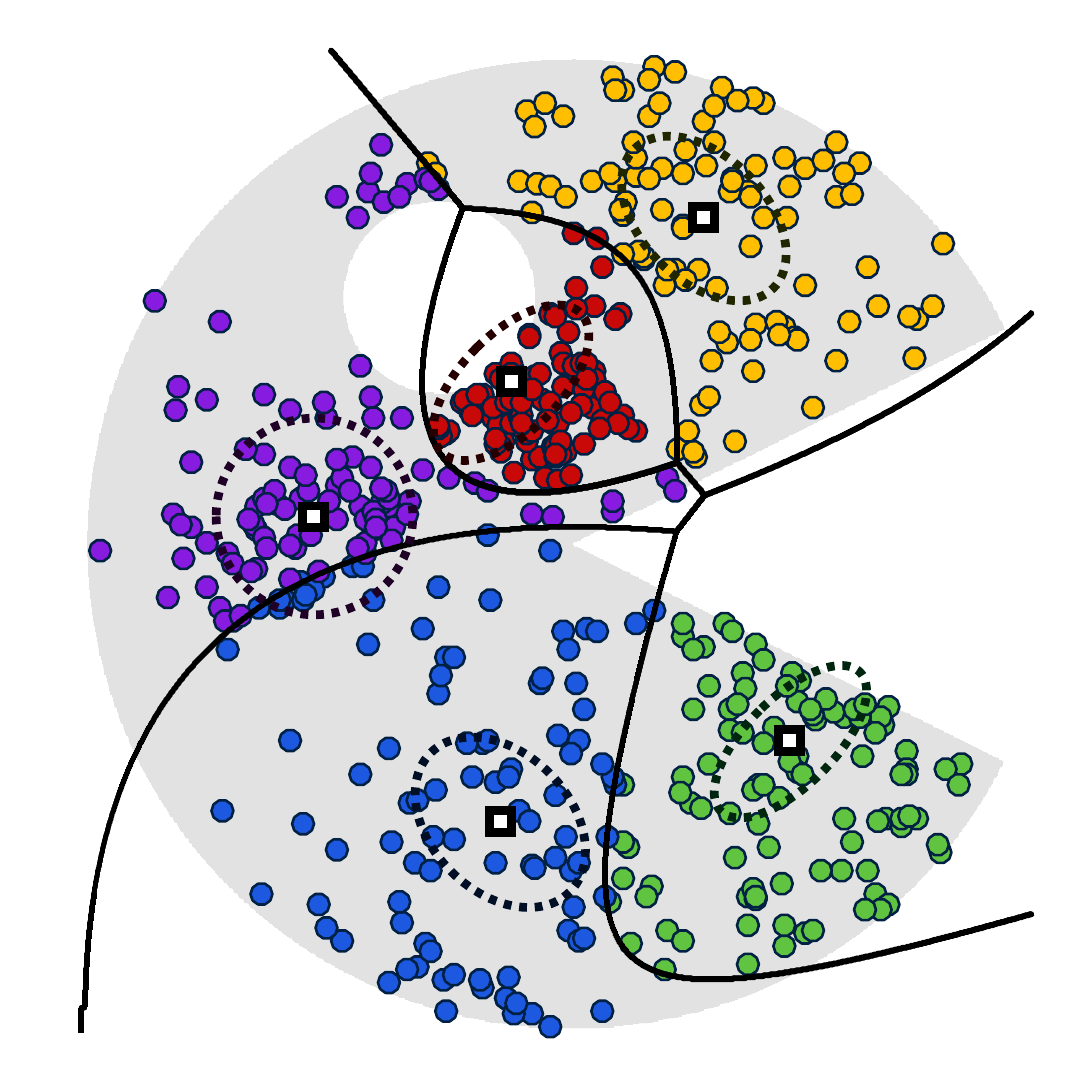}
}

\newcommand{\figureToyExSP}[1][0.43\textwidth]{
	\includegraphics[width = #1]{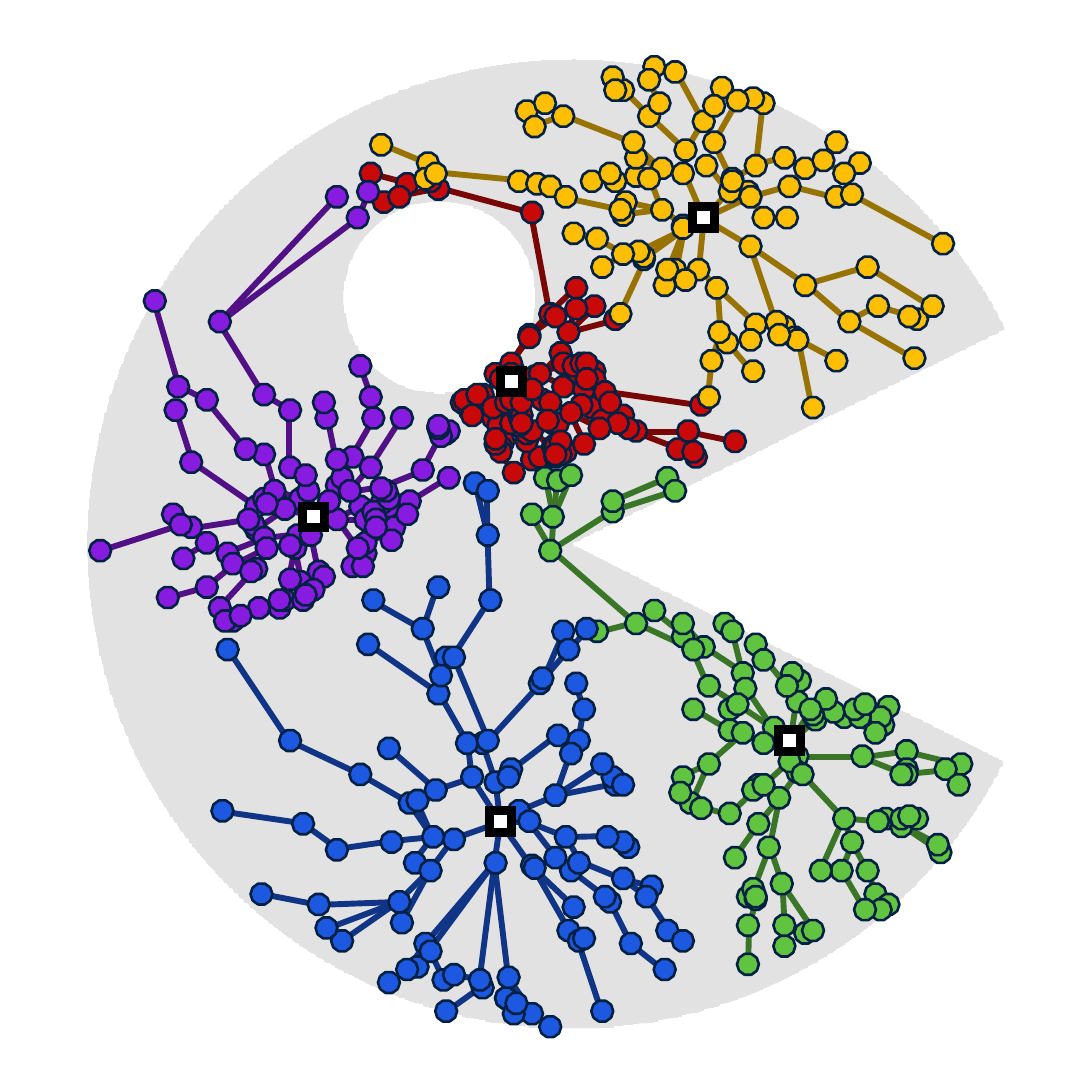}
}
\newcommand{\figureToyExSPSQ}[1][0.43\textwidth]{
	\includegraphics[width = #1]{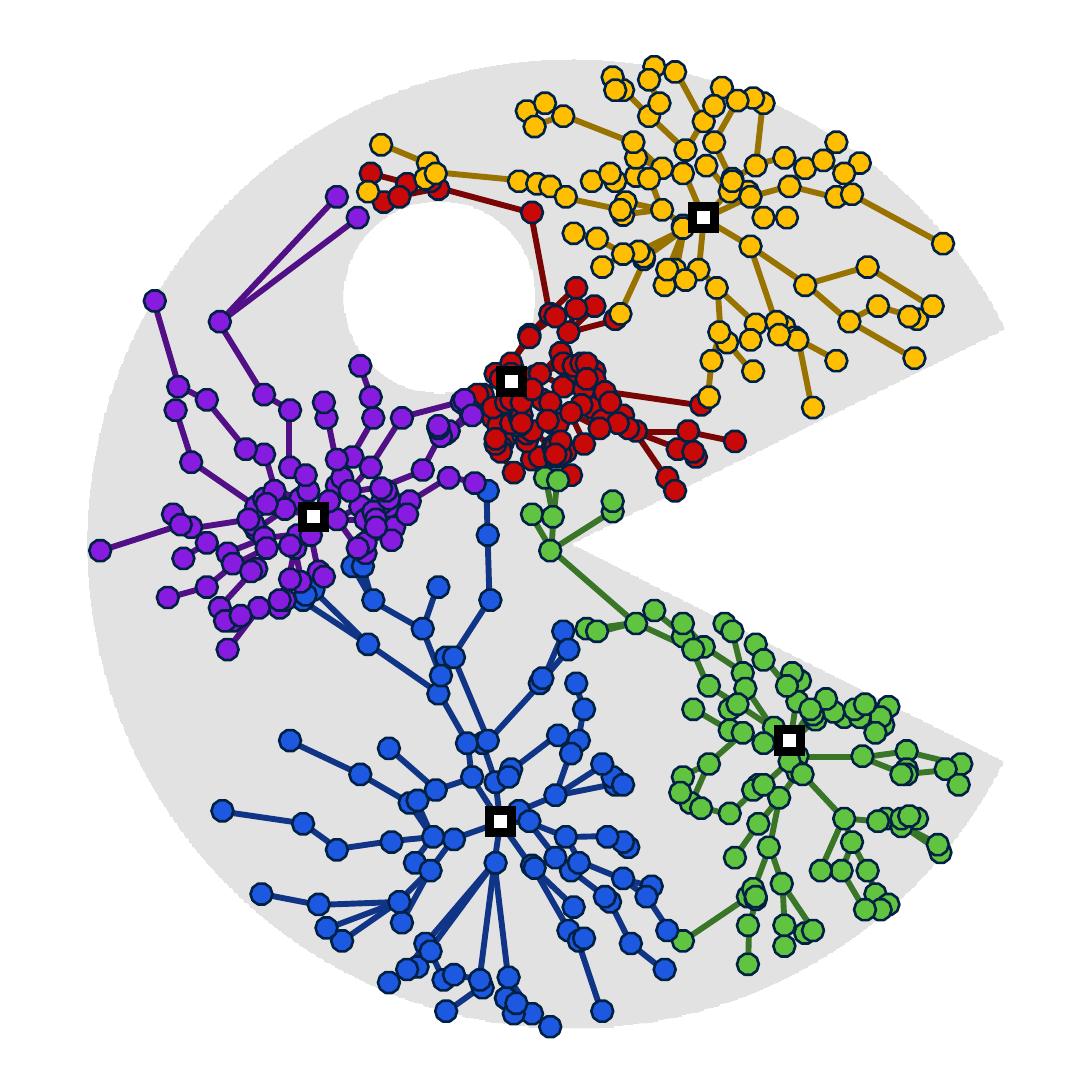}
}

\newcommand{\placeFigureA}{
 \begin{figure*}[h!t]
	\begin{subfigure}[t]{.34\textwidth}
		\centering
		\figureDevGermanyElection
		\subcaption{Deviations of the 2013 election districts in Germany}
		\label{fig:res:dev:2013}
	\end{subfigure}%
	\hspace{.02\textwidth}
	\begin{subfigure}[t]{.34\textwidth}
		\centering
		\figureDevGermanySP
		\subcaption{Deviations resulting from  our  methodology.}
		\label{fig:res:dev:shortestpath}
	\end{subfigure}%
	\hspace{.01\textwidth}
	\begin{subfigure}[t]{.26\textwidth}
		\figureDevLegend
		\subcaption{Colors depict the deviation from the federal average
			district size.}
		\label{fig:res:dev:legend}
	\end{subfigure}%
	\caption{Absolute deviations from the average population size per district.}
	\label{fig:res:dev:overview}
\end{figure*}
}

\newcommand{\placeFigureB}{
  \begin{figure}[!ht]
	\centering
	\begin{subfigure}[t]{.33\textwidth}
		\figureToyExPDDiagram
		\subcaption{Diagram in Euclidean space.}
	\end{subfigure}%
	\begin{subfigure}[t]{.33\textwidth}
		\centering
		\figureToyExAPDDiagram
		\subcaption{Anisotropic diagram with local ellipsoidal norms.}
	\end{subfigure}%
	\begin{subfigure}[t]{.33\textwidth}
		\centering
		\figureToyExSPDiagram
		\subcaption{Diagram in discrete space.}
	\end{subfigure}%
	\caption{Exemplary clusterings and related diagrams.}
	\label{fig:exemplarydissections}
\end{figure}	
}

\newcommand{\placeFigureC}{
\begin{figure}[htb] \centering
	\figureToyEx
	\caption{Exemplary constrained clustering instance of
		\cref{sec:classes}}
	\label{fig:exmpl:instance}
\end{figure}
}

\newcommand{\placeFigureD}{
\begin{figure}[!ht]
\begin{minipage}{0.48\textwidth}
	\centering
	\figureToyExAWVD[.9\textwidth]
	\caption{Clustering \wrt $\hat \dfunc_i(\unit) = ||s_i - \unit||$}
	\label{fig:exmpl:awvd}
\end{minipage} \hspace{0.02\textwidth}
\begin{minipage}{0.48\textwidth}
	\centering
	\figureToyExPD[.9\textwidth]
	\caption{Clustering \wrt $\hat \dfunc_i(\unit) = ||s_i - \unit||^2$}
	\label{fig:exmpl:powerdiagram}
\end{minipage}
\end{figure}
}

\newcommand{\placeFigureE}{
\begin{figure}[!ht]
	\begin{minipage}{0.48\textwidth}
		\centering
		\figureToyExAVD[.9\textwidth]
		\caption{Clustering \wrt $\hat \dfunc_i(\unit) = \sqrt{ (\unit\!-\!s_i)\tr M_i (\unit\!-\!s_i)}$}
		\label{fig:exmpl:anisotropvd}
\end{minipage} \hspace{0.02\textwidth}
\begin{minipage}{0.48\textwidth}
		\centering
		\figureToyExAPD[.9\textwidth]
		\caption{Clustering \wrt $\hat \dfunc_i(\unit) = (\unit\!-\!s_i)\tr M_i
			(\unit\!-\!s_i)$}
		\label{fig:exmpl:anisotroppd}
\end{minipage}
	\end{figure}
}

\newcommand{\placeFigureF}{
\begin{figure}[!ht]
	\begin{minipage}{0.48\textwidth}
		\centering
		\figureToyExSP[.9\textwidth]
		\caption{Clustering \wrt $\hat \dfunc_i(\unit) = \dgraph(s_i,\unit)$}
		\label{fig:exmpl:graphclustering}
\end{minipage} \hspace{0.02\textwidth}
\begin{minipage}{0.48\textwidth}
		\centering
		\figureToyExSPSQ[.9\textwidth]
		\caption{Clustering \wrt  $\hat \dfunc_i(\unit) = \dgraph(s_i,\unit)^2$}
		\label{fig:exmpl:graphclusteringSquared}
\end{minipage}
\end{figure}
}

\newcommand{\placeFigureG}{
\begin{figure}[htb]
	\centering
	\figureExampleOnlyAffineTrans
	\caption{Intersecting shortest paths.}
	\label{fig:proof:connectedonlyaffine}
\end{figure}	
}

\newcommand{\placeFigureH}{

	 \begin{figure}[!ht]
		\figureZoomedClusteringFracts
		\centering
		\caption{Rounded units in the case $\hat \dfunc_i(\unit) =
			\dgraph(s_i,\unit)$.}
		\label{fig:exmpl:graphclustering:fracts}
	\end{figure}
	\begin{figure}[!ht]
		\centering
		\figureZoomedClusterinSquaredConnected
		\caption{Non-connected cluster in the case $\hat \dfunc_i(\unit) =
			\dgraph(s_i,\unit)^2$.}
		\label{fig:exmpl:graphclusteringSquared:nonconnected}
	\end{figure}

}

\newcommand{\placeFigureI}{
	\begin{figure}[!ht]
		\centering
		\begin{subfigure}[t]{.48\textwidth}
			\centering
			\figureBayernOriginalDistricts
			\caption{Districts from the 2013 elections.}
			\label{fig:res:bayern:wk13}
		\end{subfigure}
		\hspace{.02\textwidth}
		\begin{subfigure}[t]{.48\textwidth}
			\centering
			\figureBayernPowerDiagram
			\caption{Districts  via power diagrams.}
			\label{fig:res:bayern:cpd}
		\end{subfigure}
		\caption{Districts for the state of Bavaria. Here and in the following, we use a  six-coloring of the districts for an easy distinguishability.}
		\label{fig:res:bayern}
	\end{figure}
}

\newcommand{\placeFigureJ}{
	\begin{figure}[!ht]
		\centering
		\begin{subfigure}[t]{.45\textwidth}
			\centering
			\figureNiedersachsenOriginalDistricts
			\caption{Districts from the 2013 election.}
			\label{fig:res:niedersachsen:wk13}
		\end{subfigure}
		\hspace{.02\textwidth}
		\begin{subfigure}[t]{.45\textwidth}
			\centering
			\figureNiedersachsenAPD
			\caption{Districts  via anisotropic power
				diagrams. Ellipses depict unit balls of local norms (determined from the districts of 2013).}
			\label{fig:res:niedersachsen:apd}
		\end{subfigure}
		\caption{Districts for the state of Lower Saxony.}
		\label{fig:res:niedersachsen}
	\end{figure}
}

\newcommand{\placeFigureK}{
	  \begin{figure}[h!t]
		\begin{subfigure}[t]{.45\textwidth}
			\centering
			\figureNRWOriginalDistricts
			\caption{Districts from the 2013 election.}
			\label{fig:res:nrw:wk13}
		\end{subfigure}
		\hspace{.02\textwidth}
		\begin{subfigure}[t]{.45\textwidth}
			\centering
			\figureNRWSP
			\caption{Districts via shortest-path
				diagrams.}
			\label{fig:res:nrw:sp}
		\end{subfigure}
		\caption{Districts for the state of  North Rhine-Westphalia.}
		\label{fig:res:nrw}
	\end{figure}
}

\newcommand{\placeFigureL}{
	\begin{figure*}[h!t]
		\centering
		\begin{subfigure}[t]{.31\textwidth}
			\centering
			\figureBundDevPD
			\subcaption{ {Power diagram} approach.}
			\label{fig:res:deviations:cpdcorrected}
		\end{subfigure}%
		\hspace{.02\textwidth}
		\begin{subfigure}[t]{.31\textwidth}
			\centering
			\figureBundDevAPD
			\subcaption{{Anisotropic power diagram} approach.}
			\label{fig:res:deviations:apd}
		\end{subfigure}%
		\hspace{.02\textwidth}
		\begin{subfigure}[t]{.31\textwidth}
			\centering
			\figureBundDevSP
			\subcaption{ {Shortest-path diagram} approach.}
			\label{fig:res:deviations:shortestpath}
		\end{subfigure}%
		\caption{Deviations after applying our methodology. Colors as in
			\cref{fig:res:dev:overview}.}
		\label{fig:res:deviations}
	\end{figure*} 
}

\newcommand{\placeFigureM}{
	 
	\begin{figure*}[!ht]
		
		\begin{subfigure}[t]{.48\textwidth}
			\centering
			\figureHessenOriginalDistricts
			\subcaption{Original districts of the 2013 election.}
			\label{fig:res:hessen:wk13}
		\end{subfigure}%
		\hspace{.02\textwidth}
		\begin{subfigure}[t]{.48\textwidth}
			\centering
			\figureHessenPD
			\subcaption{Districts via power diagrams. }
			\label{fig:res:hessen:CPDcorrected}
		\end{subfigure}\\[1em]
		
		\begin{subfigure}[t]{.48\textwidth}
			\centering
			\figureHessenAPD
			\subcaption{Districts via anisotropic power diagrams.}
			\label{fig:res:hessen:APD}
		\end{subfigure}%
		\hspace{.02\textwidth}
		\begin{subfigure}[t]{.48\textwidth}
			\centering
			\figureHessenSP
			\subcaption{Districts via shortest-path diagrams.}
			\label{fig:res:hessen:shortestpath}
		\end{subfigure}
		\caption{Districts for the state of Hesse resulting from the different
			approaches. See also \cref{tab:res:overview} for corresponding key figures.}
		\label{fig:res:hessen}
	\end{figure*}
}

\newcommand{\ie}{i.\,e.\xspace}%
\newcommand{\eg}{e.\,g.\xspace}%
\newcommand{\cf}{cf.\xspace}%
\newcommand{\wrt}{w.\,r.\,t.\xspace}%
\newcommand{\Rbb}{\mathbb{R}}
\newcommand{\Qbb}{\mathbb{Q}}
\newcommand{\Nbb}{\mathbb{N}}
\newcommand{\Rbbsp}{\Rbb_{> 0}}
\newcommand{\Rbbp}{\Rbb_{\ge 0}}

\newcommand{\NP}{{\mathcal{NP}}}
\newcommand{\ones}{\mathds{1}}
\newcommand{\tr}{\ensuremath{^{\text{\tiny T}}}}

\newcommand{\set}[1]{\ensuremath{\left\{#1\right\}}}

\newcolumntype{o}{@{}>{{}}c<{{}}@{}}
\newcolumntype{L}[1]{>{\raggedright\arraybackslash}p{#1}}
\newcolumntype{R}[1]{>{\raggedleft\arraybackslash}p{#1}}
\newcolumntype{C}[1]{>{\centering\arraybackslash}p{#1}}

\newcommand{\mX}{\mathcal{X}} 

\newcommand{\unit}{x}
\newcommand{\unitweight}{\omega}
\newcommand{\unitset}{X}

\newcommand{\cweight}{\kappa}
\newcommand{\cweightl}{\kappa^{-}}
\newcommand{\cweightu}{\kappa^{+}}

\newcommand{\cweightvector}{\mathcal{K}}

\newcommand{\uweightvector}{\Omega}

\newcommand{\conngraph}{G}
\newcommand{\dgraph}{d_G}

\newcommand{\id}{{\operatorname{id}}}

\newcommand{\voronoi}{P}

\newcommand{\voronoidg}{\mathcal{P}}

\newcommand{\muvector}{\mathcal{M}}

\newcommand{\cyclicexchange}{\mathcal{Z}}

\newcommand{\etavector}{\mathcal{E}}

\newcommand{\sitevector}{\mathcal{S}}

\newcommand{\dfunc}{f}
\newcommand{\funcset}{\mathcal{F}}
\newcommand{\funcfam}{\left( \dfunc_1,\ldots,\dfunc_k \right)}

\newcommand{\dset}{\mathcal{D}}
\newcommand{\dfam}{\left( d_1,\ldots,d_k \right)}

\newcommand{\Sfam}{\left( s_1,\ldots,s_k \right) }

\newcommand{\funcparamvector}{(\dset,h,\sitevector,\muvector)}

\newcommand{\arcset}{E}
\newcommand{\arcdistfunc}{\delta}

\newcommand{\BC}{\operatorname{BC}}
\newcommand{\BCI}{\BC_{I}}
\newcommand{\BCIeps}{\BCI^{\hspace{1pt}\epsilon}}
\newcommand{\BCeps}{\BC^{\epsilon}}

\newcommand{\cluster}{C}
\DeclareMathOperator{\supp}{supp}
\newcommand{\suppCi}{\supp(\cluster_i)}
\newcommand{\suppCl}{\supp(\cluster_l)}

\newcommand{\clustering}{\mathcal{C}}
\newcommand{\clusterweighti}{\unitweight(\cluster_i)}

\newcommand{\ccenter}{c}

\newcommand{\assigngraph}{H}

\newcommand{\xiO}{\xi^{\text{o}}}
\newcommand{\clusteringO}{\clustering^{\text{o}}}
\newcommand{\clusterO}{\cluster^{\text{o}}}

\newcommand{\singv}{\sigma}

\newcommand{\anormM}{M}
\newcommand{\anormi}[1]{||#1||_{\anormM_i}}

\newcommand{\dnorm}[1]{\left\lVert#1\right\rVert}

\newcommand{\unitball}{B}

\newtheorem{thm}{Theorem}
\newtheorem{corollary}[thm]{Corollary}
\newtheorem{lemma}[thm]{Lemma}

\title{Constrained clustering via diagrams:\\
	A unified theory and its applications to electoral district design}

\author{Andreas Brieden \thanks{ andreas.brieden@unibw.de,  Universit\"at der Bundeswehr, 85579 Neubiberg, Germany}
\and Peter Gritzmann \thanks{gritzmann@tum.de, Department of Mathematics, Technical University of Munich, 80290
	M\"unchen, Germany }
\and Fabian Klemm \thanks{klemm@ma.tum.de, Department of Mathematics, Technical University of Munich, 80290
	M\"unchen, Germany}
}

\begin{document}
	 \maketitle


		\begin{abstract}
			
The paper develops a general framework for constrained clustering
which is based on the close connection of geometric clustering and 
diagrams. Various new structural and algorithmic results are proved
(and known results generalized and unified)
which show that the approach is computationally efficient and flexible enough 
to pursue various conflicting demands. 

The strength of the model is also demonstrated practically on real-world instances of 
the electoral district design problem where municipalities of a state have to be 
grouped into districts of nearly equal population while obeying certain politically 
motivated requirements.

		\end{abstract}

\section{Introduction}
\label{sec:introduction}
 \paragraph{Constrained Clustering}
General clustering has long been known as a fundamental part of combinatorial
optimization and data analytics. 
For many applications (like electoral district design) it is, however, essential to 
observe additional constraints, particularly on the cluster sizes. 
Accordingly, the focus of the present paper is on {\em constrained clustering}
where a given weighted point set $\unitset$ in some space $\mX$ has to be partitioned 
into a given number $k$ of clusters of (approximately) {\em predefined weight}. 

As has been observed for several applications, good clusterings are
closely related to various generalizations of Voronoi diagrams;
see e.g. \cite{BG2012}, \cite{GKPR2014}, \cite{Carlsson2015} for recent work that is
most closely related to the present paper.
Besides electoral district design, such applications include
grain-reconstruction in material sciences (\cite{ABG+2015}), farmland consolidation 
(\cite{BBG2011}, \cite{BG2012}, \cite{BBG2014}), facility and service
districting (\cite{Segal1977}, \cite{Zoltners1983a},
\cite{Mulvey1984},  \cite{Muyldermans2002}, \cite{Kalcsics2002},
\cite{Galvao2006}, \cite{Aronov2009}), and robot network design
(\cite{Cortes2010}, \cite{Carlsson2013}).

We will present a general unified theory which is based on the relation of 
constrained geometric clustering and diagrams. In Sections \ref{sec:defsandresults} and 
\ref{sec:classes}, we analyze the model and prove various favorable properties. 

Using several types of diagrams  in different spaces, we obtain partitions
that are optimized with respect to different criteria: In Euclidean space, we obtain clusters
that are particularly well consolidated. Using locally ellipsoidal norms,  we
can to a certain extent preserve originally existing structures. 
In a discrete metric space derived from a graph that encodes an intrinsic neighboring relation,
we obtain assignments that are guaranteed to be connected.
In the theoretical part the various different issues will be visualized with the 
help of a running example.

\paragraph{Electoral District Design} Our prime example will
be that of electoral district design which has been approached from various 
directions over the last half century (see
\cite{RSS2013}, \cite{Kal2015}, and \cite{Tas2011} for surveys, \cite{God2015}
for the example of Germany, and \cite{HG2012}, \cite{HG2013} for general accounts
on partitions).
Municipalities of a state have to be grouped to form electoral districts.
The districts are required to be of nearly equal population and of ``reasonable'' shape.
Hence a crucial nature of the electoral district design problem is that
there are several partly conflicting optimization criteria such as the grade of population balance,
consolidation, or a desire for some continuity in the development of districts
over time. Therefore we will show how our unified approach allows the decision maker 
to compare several models with different optimization foci. 

\placeFigureA

\cref{sec:results} will show the effect of our method for the federal elections in Germany. 
The German law (\cite{bwg}) requires that any deviation of district sizes of more 
than 15\% from the federal average is to be avoided.
As a preview, \cref{fig:res:dev:overview} contrasts the occurring deviations from the 2013
election with the deviations resulting from one of our approaches.
The federal average absolute deviation drops significantly from 9.5\% for
the 2013 election to a value ranging from 2.1\% to 2.7\% depending on
the approach. For most states, these deviations are close to optimal since the 
average district sizes of the states i.e., the ratios of their numbers of districts 
and eligible voters differ from the federal average already about as much. See \cref{sec:results} for
detailed results and the Appendix for further statistics. Furthermore, an online supplement depicts the results of all approaches for the full data set, see
 \url{http://www-m9.ma.tum.de/material/districting/}.

\paragraph{Constrained Clustering via Generalized Voronoi Diagrams}
~ \newline
In accordance with \cite{ES1986}, the \emph{generalized Voronoi diagram} for given
functions $\dfunc_i: \mX \rightarrow \Rbb$, $i = 1,\ldots,k$, is obtained by
assigning each point $x \in \mX$ to a subset $C_i$ of $\mX$ whose value
$\dfunc_i(x)$ is minimal.
We are interested in clusterings of $\unitset$ that are induced by such
diagrams (cf.
\cref{sec:method:clusteringsviageneralizedvds}).

Of course, in order to obtain suitable diagrams, the choice of the functions
$\dfunc_i$ is crucial. For parameters $\funcparamvector$
we define the $k$-tuple of functions $\funcset(\dset,h,\sitevector,\muvector) 
:= \funcfam$ via \[\dfunc_i(x) := h(d_i(s_i,x)) + \mu_i.\] Here, $\dset := \dfam$ is a
$k$-tuple of metrics (or more general distance measures) on $\mX$,
 $h: \Rbbp \rightarrow \Rbbp$ is monotonically increasing, $\sitevector := \Sfam \in \mX^k$ is a $k$-tuple of points in $\mX$,
and $\muvector := (\mu_1,\ldots,\mu_k) \in \Rbb^k$ is a vector of reals.
 If the metrics $d_i$ are not all identical, we call the resulting
 diagram \emph{anisotropic}.

 \placeFigureB
 
 We consider an exemplary selection of types of generalized Voronoi diagrams (see
 also \cite{AB1986}, \cite{ES1986}, \cite{Okabe1997}, \cite{OBSC2009}). 
For each of the considered types, \cref{fig:exemplarydissections} depicts an exemplary diagram together with its induced clustering.
 
 In the Euclidean space,  the choice  
 \[\dfunc_i(x) := \dnorm{ x - s_i }_2^{2} + \mu_i\]
 yields \emph{power diagrams}; see \cite{AHA1998}, \cite{bhr-92}. For the particular case of
 \emph{centroidal} diagrams, in which the sites coincide with the resulting
 centers of gravity of the clusters, the inherent variance is minimized.
 This can be achieved by optimization over $\sitevector$ (cf. \cite{BG2012}, \cite{BG2010},
 \cite{BBG2013}, \cite{FH2011}).
 
 The setting 
 \[\dfunc_i(x) := \dnorm{ x - s_i }_2+ \mu_i\]
 yields  \emph{additively weighted Voronoi diagrams}. 
 
 Allowing for each cluster an individual ellipsoidal norm yields \emph{anisotropic}
 Voronoi and power diagrams, respectively. 
 Appropriate  choices of norms facilitate the integration of further information
 such as the shape of pre-existing clusters in our application.
 
 We also consider the discrete case $\mX = \unitset$. Here, we are
 given a connected graph  $\conngraph := (\unitset, \arcset, \arcdistfunc)$ 
 with a function $\delta: \arcset \rightarrow \Rbbsp$ assigning a positive
 distance to each edge. With $\dgraph(x,y)$  defined as the length of the
 shortest $x$-$y$-path in $\conngraph$ \wrt $\delta$, this induces a metric on
 $\mX$.
 The choice of
 \[\dfunc_i(x) := \dgraph(s_i,x) + \mu_i\]
  then leads to
 \emph{shortest-path diagrams}.
 Such diagrams guarantee the connectivity of all clusters in the
 underlying graph.
 This allows to represent intrinsic relations of data points that cannot be
 easily captured otherwise.

As we will see, the parameters $\dset$ and $h$ mainly
determine the characteristics of the resulting diagrams. 
The points $s_i$ then serve as reference points -- called \emph{sites} -- for
the clusters. 

It is shown that for any choice of $\dset$, $h$ and $\sitevector$ there exists a
choice of the additive parameter tuple $\muvector$, such that the induced
clusters are of prescribed weight as well as optimally consolidated (cf.
\cref{cor:optimalclustering}).  Thus, we distinguish between the
\emph{structural parameters}  $\dset,h$ and $\sitevector$ and the \emph{feasibility parameter}
$\muvector$. Our approach does not automatically yield integral assignments in
general but may require subsequent rounding.
However,  the number of fractionally assigned points and thus the deviation of
cluster weights can be reasonably controlled (see \cref{lem:errorestimate}).
 
Typically, $\dset$ and $h$ are defined by the specific application as it determines
the requirements on the clusters. One can still optimize
over the remaining structural parameter $\sitevector$ with respect to different criteria,
\eg, optimal total variances or margins. For any choice of structural
parameters, the feasibility parameter $\muvector$ is then readily provided as
the dual variables of a simple linear program.

As we will point out in more detail our framework extends various previously pursued approaches.
We feel that a unified and self-contained exposition serves the reader 
better than a treatment that relies heavily on pointers to the scattered
literature. Hence we include occasionally new concise proofs of known
results whenever this adds to the readability of the paper. 
Of course, we try to always give the appropriate references. 

 \paragraph{Organization of the Paper}
\cref{sec:defsandresults} yields the general definitions and methodology for our approach. 
\cref{sec:classes} provides a short study of typical generalized Voronoi diagrams 
and shows their relevance for constrained clustering.
\cref{sec:results} then presents our results for the electoral district design
problem for the example of Germany in all detail, while \cref{sec:conclusion}
concludes with some final remarks.

\section{Definitions and Methodology}
\label{sec:defsandresults}
We begin by describing our approach to constrained geometric clustering in a general context.  
Due to the specific application we focus on the discrete case of partitioning a given
finite weighted set; some results for the continuous case will however also be mentioned.

First, \cref{sec:method:constraintedclustering} defines the terminology for
constrained clusterings. We construct  clusterings  that are
induced by a suitable dissection of  the underlying space. For this purpose,
\cref{sec:method:clusteringsviageneralizedvds} formally defines  generalized
types of Voronoi diagrams and shows how they relate to clusterings.
\cref{sec:method:clusteringdiagramcorrespondence} then yields the main
theoretical results that establish a correspondence of clusterings with
prescribed capacities and  generalized Voronoi diagrams.

\subsection{Constrained Clustering} \label{sec:method:constraintedclustering}
Let $k,m \in \Nbb$ and $\mX$ be an arbitrary space.  We consider a  set
\[\unitset := \{\unit_1,\ldots,\unit_m\} \subset \mX\] 
with
corresponding weights
 \[\uweightvector := (\unitweight_1,\ldots,\unitweight_m)
\in \Rbbsp^m.\]
Furthermore, let 
\[\cweightvector := (\cweight_1,\ldots,\cweight_k) \in
\Rbbsp^k\] 
such that $\sum_{i = 1}^k \cweight_i = \sum_{j=1}^m \unitweight_j$.

The vector $\cweightvector$ contains the desired cluster "sizes". Hence, we want to find a partition of $\unitset$ such that for each cluster $C_i$ its
total weight meets the prescribed capacity $\cweight_i$.

For $k = 2$, integer weights, and $\cweight_1 = \cweight_2$ the associated
decision problem coincides with the well-known \textsc{Partition} problem
and is therefore already $\NP$-hard.

We  consider also a relaxed version of the problem by allowing fractional
assignments 
\[\clustering := \left(   \xi_{i,j} \right)_{ \substack{ i = 1,\ldots,k
\\ j = 1,\ldots,m }} \in [0,1]^{k \times m}\]
such that $\sum_{ i = 1}^k
\xi_{i,j} = 1$ for each $j$. $\clustering$ is called a \emph{(fractional)
clustering} of $\unitset$ and $\xi_{i,j}$ is the fraction
of unit $j$ assigned to cluster $i$. We further set  $\cluster_i :=
(\xi_{i,1},\ldots,\xi_{i,m})$, call it cluster $i$ and let
\[\suppCi := \{ \unit_j \in \unitset: \xi_{i,j} > 0 \} \]
 denote its \emph{support}, \ie,  the set of those
elements in $\unitset$ that are assigned to $i$ with some positive fraction. If
$\clustering \in \{ 0,1 \}^{k \times m}$, we call the clustering \emph{integer}.

The \emph{weight} of a cluster is given by 
\[\clusterweighti := \sum_{j = 1}^m \xi_{i,j} \unitweight_j .\] 
A clustering $\clustering$ is \emph{strongly balanced}, if 
\[\clusterweighti = \cweight_i\]
for each $i$.
If lower and upper bounds $\cweightl_i, \cweightu_i \in \Rbbp$ for the cluster weights are given
and
\[\cweightl_i \leq \sum_{j = 1}^m \unitweight_j \xi_{i,j} \leq \cweightu_i\]
 holds for every $i$, $\clustering$ is called \emph{weakly} balanced.
A case of special interest for our application is that of 
$\cweightl_i=(1-\epsilon)\cweight_i$ and $\cweightu_i=(1+\epsilon)\cweight_i$
 for all $i$ for some given $\epsilon > 0$.
Then, i.e., if 
\[ ( 1-\epsilon) \cweight_i  \leq \clusterweighti \leq (1+  \epsilon) \cweight_i\]
 for each 
$i$ we call $\clustering$ \emph{$\epsilon$-balanced}, or, whenever the choice of $\epsilon$ is clear simply
\emph{balanced}.

By $\BC$ and $\BCeps$ we denote the set of all strongly balanced and
$\epsilon$-balanced fractional clusterings, respectively.
Note that the condition $\sum_{i = 1}^k \cweight_i = \sum_{j=1}^m \unitweight_j$
guarantees that $\BC \neq \emptyset$.
Similarly, let $\BCI$ and $\BCIeps$ denote the set of all strongly balanced and
$\epsilon$-balanced integer clusterings, respectively.
 Of course, \linebreak
  $\BCI \subset \BC \subset \BCeps$ and $\BCI \subset \BCIeps \subset \BCeps$.

\subsection{Clusterings induced by Generalized Voronoi Diagrams}
\label{sec:method:clusteringsviageneralizedvds}

 Let a $k$-tuple $\funcset := \funcfam$ of functions $\dfunc_i : \mX
\rightarrow \Rbb$ for $i = 1,\ldots,k$ be given. For each cluster, the
corresponding $\dfunc_i$ is supposed to act as a \emph{distance measure}.
 While a  classical Voronoi diagram in Euclidean space assigns each point to a 
 reference point which is closest, this concept can be generalized by assigning
 a point to each region for which the corresponding value of $\dfunc_i$ is
 minimal.
Formally, we set
\begin{align*}
 	\voronoi_i := \{ x \in \mX: \dfunc_i(x) \leq \dfunc_l(x)~ \forall l
 	\in \{ 1,\ldots,k \}  \} 
 \end{align*}
and call $\voronoi_i$ the $i$-th \emph{(generalized) Voronoi region} (or
\emph{cell}).
Then $\voronoidg := (\voronoi_1,\ldots,\voronoi_k)$ is the
\emph{generalized Voronoi diagram} (\wrt $\funcset$). 

Note that, in general, $\voronoidg$ does not constitute a partition of $\mX$. We
will, of course, in the application  focus on choices of the functions
$\dfunc_i$ that guarantee that the cells $\voronoi_i$ do not have interior
points in common; see \cref{thm:lowdimbisectors}.

A generalized Voronoi diagram $\voronoidg$ is said to be \emph{feasible}
for a clustering $\clustering$ if 
\[\suppCi \subset \voronoi_i\] 
for all $i$.
Typically, we do not want a Voronoi region to contain elements ``by chance",
\ie, points that are not (at least fractionally) assigned to their corresponding
cluster.
Hence, we say $\voronoi$ \emph{supports} $\clustering$, if 
\[\suppCi = \voronoi_i \cap \unitset\]
 for all $i$.

\subsection{Correspondence of Constrained Clusterings and Generalized Voronoi
Diagrams}
\label{sec:method:clusteringdiagramcorrespondence}

As described in the introduction, we are interested in finding a clustering
$\clustering \in \BC$ that is supported by a generalized Voronoi diagram
$\voronoidg$ \wrt functions $\dfunc_i(x) := h(d_i(s_i,x)) + \mu_i$. A
natural question is for which choices of $\funcparamvector$ such a clustering exists.

   By definition, a diagram $\voronoidg$ is feasible for $\clustering \in \BC$
 if and only if
 \begin{align}
 \xi_{i,j} \cdot  \left( h(d_i(s_i,\unit_j)) + \mu_i  -
 \min_{l = 1,\ldots,k} ( h(d_l(s_l,\unit_j)) + \mu_l)  \right)= 0 \label{eq:dgfeasibility}
 \end{align}
 holds for every $i = 1,\ldots,k$ and $j =  1,\ldots,m$.

 We will now recover \eqref{eq:dgfeasibility} as a complementary slackness condition in linear programming.
 For this purpose, first note that in general, \ie, for any $\clustering \in \BC$ and
 $\funcparamvector$, we have
  \begin{multline}
 	\sum_{i = 1}^k \sum_{j = 1}^m  \xi_{i,j}\cdot \unitweight_j \cdot  \left(
 	h(d_i(s_i,\unit_j)) + \mu_i - \min_{l = 1,\ldots,k} ( h(d_l(s_l,\unit_j)) +
 	\mu_l)  \right) \geq 0,   \label{eq:dginequality} 
 \end{multline}
 as all weights $\unitweight_j$ are
 positive and each factor in the sum above is non-negative. 
 Using   $\sum_{j = 1}^m
 \unitweight_j \xi_{i,j} =  \cweight_i$ for each $i$ and $\sum_{i = 1}^k \xi_{i,j} = 1$ for each $j$, Inequality
 \eqref{eq:dginequality} is equivalent to
 \begin{multline}
 \sum_{i = 1}^k \sum_{j = 1}^m  \xi_{i,j}\cdot \unitweight_j \cdot  
 h(d_i(s_i,\unit_j)) 
 \geq	\sum_{j = 1}^m \unitweight_j \min_{l =
 	1,\ldots,k} ( h(d_l(s_l,\unit_j)) + \mu_l) -   \sum_{i = 1}^k \cweight_i \mu_i.
 \label{eq:dginequalityII}
 \end{multline}
 Note that  \eqref{eq:dgfeasibility} holds for every $i$ and $j$ if and only if  \eqref{eq:dginequality},
  and hence \eqref{eq:dginequalityII}, holds with equality.

 Now, observe that the left-hand side of \eqref{eq:dginequalityII} does not depend on
 $\muvector$ while the right-hand side does not depend on $\clustering$.
For any choice of $(\dset,h,\sitevector)$ equality in
 \eqref{eq:dginequalityII} can therefore only hold if $\clustering$ minimizes
 the left-hand side while  $\muvector$  maximizes the
 right-hand side.
 Thus, in particular, $\clustering \in \BC$ must be a minimizer of the linear program

  {\renewcommand{\arraystretch}{1}
\begin{equation}  \tag{P} \label{plp}
  \begin{array}{p{20pt}>{\displaystyle}r@{}>{\displaystyle}l@{\hspace{5pt}}o@{\hspace{5pt}}r@{\hspace{15pt}}>{\displaystyle}l}
  \multicolumn{6}{>{\displaystyle}l}{ \min\limits_{\clustering \in \Rbb^{k
  \times m}}  \sum_{i = 1}^k \sum_{j = 1}^m \xi_{i,j}\cdot \unitweight_j \cdot h(d_i(s_i,\unit_j))  ~~  \text{ s.t.
	 } } \\
 	   & \sum_{i = 1}^k \xi_{i,j}&  & =  &  1 & (i =
 	   1,\ldots,k) \\
 	    & \sum_{j = 1}^m \xi_{i,j}& \unitweight_j & =  &  \cweight_i & (j =
 	    1,\ldots,m) \\
 	   & \xi_{i,j}& & \geq & 0 & (i =
 	   1,\ldots,k; j = 1,\ldots,m).
  \end{array}
\end{equation}}

 By introducing auxiliary
 variables $\etavector := (\eta_1,\ldots,\eta_m)$, maximization of the
 right-hand side of  \eqref{eq:dginequalityII} can be formulated as  linear program, as
 well:
 {\renewcommand{\arraystretch}{1.1}
\begin{equation} \tag{D} \label{dlp}
  \begin{array}{p{7pt}>{\displaystyle}r@{\hspace{2pt}}o@{\hspace{2pt}}>{\displaystyle}r@{\hspace{5pt}}>{\displaystyle}l}
  \multicolumn{5}{>{\displaystyle}l}{\max_{{\substack{\muvector \in
	 \Rbb^k, \\
	\etavector \in \Rbb^m}} } \hspace{.2em} \sum_{j = 1}^m \unitweight_j  \eta_j  -  
	\sum_{i = 1}^k \cweight_i  \mu_i  ~~  \text{ s.t.
	 } } \\
 	   &  \eta_j  & \leq &  h(d_i(s_i,\unit_j)) + \mu_i & (i = 1,\ldots,k; j =
       1,\ldots,m)
  \end{array}
\end{equation}}
Now, observe that \eqref{dlp} is  the dual program to
\eqref{plp}.
Thus, $\voronoidg$ is feasible for $\clustering$ if and only if $\clustering$ and
$(\muvector, \etavector)$ are primal-dual optimizers.
In particular, as
 \[\eta_j = \min_{i = 1,\ldots,k} ( h(d_l(s_i,\unit_j)) +
 \mu_i)\]
 holds for every optimal solution of \eqref{dlp}, \cref{eq:dgfeasibility} states exactly the complementary slackness conditions.
 Furthermore, if the complementarity is strict, \ie, exactly one factor is
 strictly positive in any  equation of type \eqref{eq:dgfeasibility}, this is equivalent to $\voronoidg$ supporting
$\clustering$.

We summarize these observations in the following theorem.

\begin{thm} \label{thm:complementaryslackness}
Let $\clustering\in \BC$, $\dset$ be a $k$-tuple of metrics on $\mX$,
$\sitevector \in \mX^k$, $h: \Rbbp \rightarrow \Rbb$, $\muvector \in \Rbb^k$,
and let $\voronoidg$ be the generalized Voronoi diagram \wrt
$\funcset(\dset,h,\sitevector,\muvector)$. Further, set $\eta_j :=  \min_{i = 1,\ldots,k}
(h(d_i(s_i,\unit_j))+ \mu_i)$ for $j = 1,\ldots,m$, and $\etavector :=
(\eta_1,\ldots,\eta_m)$,

Then $(\muvector, \etavector)$ is feasible for \eqref{dlp} and the following
equivalencies hold:

{\renewcommand{\arraystretch}{1.75}
\begin{tabular}{L{0.24\textwidth}@{\hspace{9pt}}l@{\hspace{9pt}}L{0.45\textwidth}}
$\voronoidg$ is feasible for $\clustering$ & $\Leftrightarrow$ &
$\clustering$ and $(\muvector, \etavector)$ satisfy the complementary slackness
condition for
\eqref{plp} and \eqref{dlp}
\\
$\voronoidg$ supports $\clustering$ & $\Leftrightarrow$ &
$\clustering$ and $(\muvector, \etavector)$  satisfy the \newline strict complementary
slackness condition for
\eqref{plp} and \eqref{dlp}
\end{tabular}
}
\end{thm}

\cref{thm:complementaryslackness} establishes a one-to-one-correspondence of
(fractional) strongly balanced clusterings that are supported by a generalized Voronoi diagram 
and those faces of the
polytope  $\BC$ which are optimal \wrt a corresponding objective function. 

Observe that the deduction of \cref{thm:complementaryslackness} did not use any
further assumptions on the functions $\dfunc_i$ besides the additive component
$\muvector$. 
Thus, we obtain the following corollary.

\begin{corollary} \label{cor:optimalclustering}
Let  $\hat \dfunc_i: \mX \rightarrow \Rbb$, $i = 1,\ldots,k$. 
Then
$\clustering^\ast \in \BC$ is an optimizer of 
$\min_{ \clustering \in \BC} 
\sum_{i = 1}^k \sum_{j = 1}^m   \xi_{i,j}\cdot \unitweight_j \cdot \hat f_i(\unit_j)$ if
and only if there exists $\muvector := (\mu_1,\ldots,\mu_k) \in \Rbb^k$ such 
that the generalized Voronoi diagram \wrt $f_i := \hat f_i + \mu_i$ is feasible
for $\clustering^\ast$.
\end{corollary}

For functions $\hat \dfunc_i(x) := h(d_i(s_i,x))$,
 this  has already been shown  by linear programming duality in \cite{BG2012} for a discrete set $X$, $h = (\cdot)^2$, and the Euclidean metric. In a continuous setting, \ie,
for $X = \Rbb^n$ and balancing constraints defined \wrt a probability
distribution on $\Rbb^n$, this has been proven in \cite{AHA1998} and
extended to more general function tuples $\funcset$ in \cite{GKPR2014}.
Here, the particular case $h = \id$  is contained in \cite{Aronov2009}.
In \cite{Cortes2010}, this result was deduced for the Euclidean case and an arbitrary function $h$ by carefully considering the
optimality conditions of an alternative optimization problem and deducing optimality (but not
establishing the linear programming duality). Furthermore, in \cite{Carlsson2013}
and \cite{Carlsson2015}  the general continuous case was proved by discretization also involving  linear programming duality.

Using the second part of \cref{thm:complementaryslackness}, we can now characterize
supporting diagrams. 

\begin{corollary}
\label{thm:supportingclustering} 
Let  $\hat \dfunc_i: \mX \rightarrow \Rbb$, $i = 1,\ldots,k$. Then 
 $\clustering^\ast \in \BC$ lies in the relative interior of the optimal
face of 
 $\min_{ \clustering \in \BC}  \sum_{i =
1}^k \sum_{j = 1}^m  \xi_{i,j}\cdot \unitweight_j \cdot
\hat \dfunc_i(\unit_j)$  
if and only if  there exists  $\muvector := (\mu_1,\ldots,\mu_k) \in \Rbb^k$ such that the 
generalized Voronoi diagram  \wrt $f_i := \hat f_i + \mu_i$ supports $\clustering^\ast $.
\end{corollary}

Thus, for non-unique optima the supporting property of the diagrams may
still be established but comes with the price of more fractionally assigned elements (cf. \cref{sec:shortestpathdiagrams}).

Another fact that roots in a basic result about extremal points of transportation
polytopes has been noted in their respective context by several authors (\eg, 
\cite{Segal1977}, \cite{Marlin1981a}, \cite{Zoltners1983a}, \cite{Hoj1996}, \cite{GLW1997}, \cite{KNS2005}, \cite{BG2012}):
An optimal basic solution of \eqref{plp} yields partitions with a limited number 
of fractionally assigned points. 

Our proof of the following Lemma \ref{cor:fractcount} relies on the bipartite \emph{assignment graph} $\assigngraph(\clustering)$
that is associated with a given clustering $\clustering \in \BC$. It is defined by 
\[\assigngraph(\clustering) := ( \{ 1,\ldots,k \} \cup \unitset, E)\] 
with
 \[ E := \{ \{i, \unit_j\} : \xi_{i,j} > 0 \}.\] 
By \cite[Thm. 4]{KW1968} $\assigngraph(\clustering)$ is acyclic if and only if $\clustering$ is extremal, \ie, a vertex of the feasible region of \eqref{plp}.

\begin{lemma} \label{cor:fractcount}
Let $\hat f_i: \mX \rightarrow \Rbb$, $i = 1,\ldots,k$.
Then there exists $\muvector := (\mu_1,\ldots,\mu_k) \in \Rbb^k$ and a clustering $\clustering \in \BC$
with at most $(k-1)$ fractionally assigned points and at most $2(k-1)$
fractional components $\xi_{i,j}$ such that the generalized Voronoi diagram \wrt
$f_i := \hat f_i + \mu_i$  is feasible for $\clustering$.
\end{lemma}

\begin{proof}
Let $\clustering$ be an extremal solution of \eqref{plp} and \linebreak $\assigngraph(\clustering) = ( \{ 1,\ldots,k \}
\cup \unitset, E)$ be the corresponding assignment graph.  As
$\assigngraph(\clustering)$ is acyclic it follows that 
 $|E| \leq k + m - 1$.
Further, from the definition of $\BC$ it follows that  $\deg(\unit_j ) \geq 1 $ for every $j = 1,\ldots,m$. Moreover, for
any fractionally assigned element $\unit_j \in \unitset$ it follows that $\deg(
\unit_j ) \geq 2 $.
As $\assigngraph(\clustering)$  is bipartite, we also have that $ |E| = \sum_{j
= 1}^m \deg(\unit_j)$.
In conclusion, this yields
\begin{align*}
 k + m - 1 &\geq |E| = \sum_{j = 1}^m  \deg(\unit_j) 
 \geq m + \frac{1}{2} \left| \{ \{ i,x_j \} : 0 < \xi_{i,j} <  1 \} \right| \\
 &\geq m + \left| \{ \unit_j \in \unitset: \unit_j \text{ is fractionally
 assigned} \} \right|,
 \end{align*}
which implies the assertion.
\end{proof}

Note that for the special case of  $\unitweight_j = 1$ for all $j$ every
extremal solution of the transportation problem \eqref{plp} is integral, \ie, $\BC =
\BCI$ (cf. \cite[Corollary 1]{KW1968}). In general, however, this is not the
case. 

Anyway, by solving the linear program \eqref{plp} and applying suitable
rounding, we obtain an integer clustering that satisfies an a-priori (and for
many applications very reasonable) error bound. 
More precisely, \cite{Schroder2001} showed that minimizing the maximum error after
rounding can be done in polynomial time using a dynamic programming approach,
while minimizing the sum of absolute or squared errors is $\NP$-hard.
In \cite{Hoj1996} it was furthermore shown that minimizing the number of fractional assignments while obeying a pre-defined
error tolerance is $\NP$-hard.

The following theorem provides an upper bound for $\epsilon$ that guarantees 
the existence of an $\epsilon$-balanced integer clustering.

\begin{thm} \label{lem:errorestimate} 
Let  $\hat \dfunc_i: \mX \rightarrow \Rbb$, $i = 1,\ldots,k$ and \linebreak $\epsilon \geq \max_{1
\leq j \leq m, 1 \leq i \leq k}  \frac{ \unitweight_j } { \cweight_i }$.
Then there exists $\clustering \in \BCIeps$ together with $\muvector :=
(\mu_1,\ldots,\mu_k) \in \Rbb^k$, such that the
generalized Voronoi diagram \wrt $f_i := \hat f_i + \mu_i$  is feasible for $\clustering$.

For rational input, \ie, $\uweightvector \in \Qbb^m_{>0}, \cweightvector \in
\Qbb^k_{>0}$, $\hat \dfunc_i(\unit_j) \in \Qbb$, $i =
1,\ldots,k$ and $j = 1,\ldots,m$,  
$\clustering$ and $\muvector$ can be determined in polynomial time.
\end{thm}

\begin{proof}
 By \cref{cor:optimalclustering} a solution $
 \tilde \clustering \in \BC$ of
 \eqref{plp} with a corresponding dual solution $(\mu_1,\ldots,\mu_k) \in
 \Rbb^k$ yields a  generalized Voronoi diagram \wrt $f_i := \hat f_i + \mu_i$
 that is feasible for $\tilde \clustering$. We may furthermore choose $\tilde \clustering$ 
 to be extremal. 
  We can also assume that the assignment graph $\assigngraph(\tilde
  \clustering)$ is a tree (otherwise we consider its connected components). 
  We may further root this tree in the node $1$ (that corresponds to cluster $\cluster_1$) and 
  w.l.o.g. assume that the other clusters are labelled according to their
  distance from the root, \ie, for $1 \leq i < l \leq k$  it holds that
  the length of a shortest $1$-$i$-path is at most the length of a shortest
  $1$-$l$-path in $\assigngraph(\clustering)$.
  
 For   each $i = 1,\ldots,k$ we  denote the index sets of units that are either non-integrally or integrally assigned to cluster  $\cluster_i$  by
  \begin{equation*}
  	\begin{array}{ccc}
  	F(i) := \{ j: 0 < \tilde
  	\xi_{i,j} < 1 \}
  	&\text{ and }&
  	I(i) :=  \{ j: \tilde \xi_{i,j} = 1 \},
  	\end{array}
  \end{equation*}
 respectively.
Now, we  construct an integer clustering $\clustering^\ast \in \{0,1\}^{k \times
m}$ in which all integral assignments are preserved, \ie, $\xi_{i,j}^\ast :=
\tilde \xi_{i,j}$ for $\tilde \xi_{i,j} \in \{0,1\}$.
  The remaining components will be determined successively for each cluster.
  
  Let us begin with $\cluster_1^\ast$. 
   We round up fractional assignments $\xi_{i,j}, j \in F(1)$, successively as long as this is allowed by the upper bound $\kappa_1$. More precisely, let 
 \[	j^\ast := 
  		\max \{j \in F(1):  
  			\sum\limits_{ \substack{  r \in F(1): \\
  						r \leq j}}  
  			\unitweight_r \leq \cweight_1 
  			-  
  			\sum\limits_{ r \in I(1)} \unitweight_r \} \]
  if the set on the right-hand side is non-empty, otherwise set $j^\ast := 0$.
  With
  \begin{equation*}
  	\xi^\ast_{1,j} :=
  	  \begin{cases}
  		1, &\text{if } j \leq j^\ast \\
  		0, &\text{otherwise.}
  	\end{cases}
  \end{equation*}   
  for every $j \in F(1)$ it then follows that  $|\sum_{j =   1}^m \xi^\ast_{{1},j} \unitweight_j - \kappa_1| <  \max_{j = 1,\ldots,m}  \unitweight_j$.

  Now let $i_0 \geq 2$ and assume that we have determined
  $\cluster^\ast_i$ appropriately for every $i <  i_0$.
 Let $x_{j_0} \in \unitset$ be the  predecessor of the node $i_0$ in the
 rooted tree $\assigngraph(\tilde \clustering)$. By assumption,  the assignment $\xi_{i,j_0}^\ast$ of unit $x_{j_0}$ to cluster $\cluster_{i}$ has already been determined for every $i < i_0$.
 We then set
   \begin{equation}
  \xi_{i_0,j_0}^\ast := \begin{cases} 
    1,&  \text{if }  \xi_{i,j_0}^\ast = 0 \text{ for all } i < i_0
    \\
    0, &\text{otherwise.}
  \end{cases}  \label{eq:roundingassignmentxj}
 \end{equation}
 In analogy to the first step, we define
  \[	j^{\ast \ast} := 
 \max \{j \in F(i_0) \setminus \{ j_0 \}:  \hspace{-.5em}
 \sum\limits_{ \substack{  r \in F(i_0): \\
 		r \leq j}}  \hspace{-.5em}  
 \unitweight_r \leq \cweight_{i_0} 
 -  
 \sum\limits_{ r \in I(i_0)} \unitweight_r  
 -  \xi_{i_0,j_0}^\ast \unitweight_{j_0}  \} \]
if the set on the right-hand side is non-empty, $j^{\ast \ast} := 0$ otherwise, and   set
  \begin{equation*}
	\xi^\ast_{i_0,j} :=
	\begin{cases}
		1, &\text{if } j \leq j^{\ast \ast} \\
		0, &\text{otherwise}
	\end{cases}
\end{equation*} for every $j \in F(i_0) \setminus \{j_0\}$ .

Every point that is
  fractionally assigned in $\tilde \clustering$ is  assigned either to its  predecessor 
   in $\assigngraph(\tilde \clustering)$ or to exactly one successor by
  \eqref{eq:roundingassignmentxj}. Thus, it holds that $\clustering^\ast \in
  \BCIeps$. As $\supp(\cluster_i^\ast) \subset \supp(\tilde \cluster_i)$, the
  already obtained generalized Voronoi diagram remains feasible for $\clustering^\ast$.
  
 Hence, under the additional rationality assumptions, we can solve the linear
 program \eqref{plp} and perform the successive rounding  in polynomial time.
\end{proof}

Note that the bound of \cref{lem:errorestimate}  is asymptotically worst-case
optimal (\eg, for $m = 1$, $\unitweight_1 = k$, $\kappa_i = 1$ for all
$i$ and letting $k \rightarrow \infty$).

As in \cite{BG2012}, we may also consider weakly balanced clusterings
where lower and upper bounds $\cweightl_i, \cweightu_i \in \Rbbp$
for the cluster weights are provided. Of course, a minimizer $\clustering^\ast$ 
over the polytope of such weakly balanced clusterings 
 is a minimizer of \eqref{plp} when setting $\cweight_i :=
 \unitweight(\cluster_i^\ast)$ for every $i$. Hence, 
  optimal weakly balanced clusterings still yield 
 supporting generalized Voronoi diagrams. Unfortunately, the converse is not true in
 general; see \cite{BG2012} for an example.

Naturally, we prefer generalized Voronoi diagrams that support a
clustering over diagrams that are only feasible, as they exclude the possibility of points
lying in a diagram's region only ``by coincidence". At the same time, we
prefer clusterings with only few fractional assignments as provided by
\cref{lem:errorestimate}. As a consequence of \cref{thm:supportingclustering},
this coincides whenever the optimum of \eqref{plp} is unique.
 In case of $\mX = \Rbb^n$, this can be guaranteed by an appropriate choice of
 the structural parameters $\dset, h, \sitevector$. 
We first show that under mild assumptions the generalized Voronoi-cells behave
properly.

\begin{lemma} \label{thm:lowdimbisectors}
Let $\mX = \Rbb^n$, $\dset := \dfam$ be a family of metrics induced by strictly
convex norms, $h: \Rbb \rightarrow \Rbb$ injective, $\sitevector :=
(s_1,\ldots,s_k) \in (\Rbb^{n})^k$ such that $s_i \neq s_l$ for $i \neq
l$, and $\muvector \in \Rbb^k$.

Then for the generalized Voronoi diagram $\voronoidg$ \wrt
$\funcset(\dset,h,\sitevector,\muvector)$ we have
$\operatorname{int}(\voronoi_i) \cap \operatorname{int}(\voronoi_l) = \emptyset$
whenever $i \neq l$.
\end{lemma}

\begin{proof}
Let $\unitball_1, \unitball_2 \subset \Rbb^n$ be the unit balls of the norms
that induce $d_1$ and $d_2$, respectively. Furthermore, denote by
$\dnorm{\cdot}_{\unitball_i}$, $i = 1,2$, the corresponding norms, \ie, 
$d_i(x,0) = \dnorm{ x }_{\unitball_i} = \min\{ \rho \geq 0: x \in \rho
\unitball_i \}$ for every $x \in \Rbb^n$, $i = 1,2$.

Suppose that there exists $x_0 \in \Rbb^n$ and $\delta > 0$
such that $x_0 + \delta \mathbb{B}_2 \subset \voronoi_1 \cap \voronoi_2$
where $\mathbb{B}_2$ is the Euclidean unit ball.
 This means  we have
  \begin{align}
 	h(  \dnorm{ x - s_1 }_{\unitball_1} ) - h(  \dnorm{ x - s_2 }_{\unitball_2} )
 	=  \mu_2 - \mu_1 \label{eq:auxeq}
 \end{align}
 for all $x \in x_0 + \delta \mathbb{B}_2$.
W.l.o.g. let
$s_1, s_2$ and $x_0$ be affinely independent.

Next, let $a \in \Rbb^n \setminus \set{0}$ such that 
\[H^{\leq}_{a,a \tr x_0} := \{ x \in \Rbb^n: a \tr x \leq a \tr x_0 \}\]
is a halfspace that supports $s_1 + \dnorm{ x_0 - s_1 }_{\unitball_1}
\unitball_1 $ in $x_0$.

If $H^{\leq}_{a,a \tr x_0}$ does not support $s_2 + \dnorm{ x_0 - s_2
}_{\unitball_2} \unitball_2 $ in $x_0$, it  follows that there exists $z \in x_0 + \delta
\mathbb{B}_2$ with $ \dnorm{ z - s_1 }_{\unitball_1} = \dnorm{ x_0 - s_1
}_{\unitball_1}$ and $ \dnorm{ z - s_2 }_{\unitball_2} \neq \dnorm{ x_0 - s_2
}_{\unitball_2}$. As $h$ is injective, this implies that \eqref{eq:auxeq} does
not hold for $z$, a contradiction.
Hence,  $H^{\leq}_{a,a \tr x_0}$ must support $s_2 + \dnorm{ x_0 - s_2
}_{\unitball_2} \unitball_2 $ in $x_0$.

Now there exist $\lambda > 1$ and $\nu \in \Rbb$ such that  $x_1 :=  s_1 +
\lambda( x_0 - s_1)$,  $x_2 := s_2 + \nu(x_0 - s_2) \in \operatorname{int}( x_0 + \delta
\mathbb{B}_2 )$ and  $\dnorm{ x_1 - s_1
}_{\unitball_1} = \dnorm{ x_2 - s_1
}_{\unitball_1}$.
Furthermore, due to the affine independence of  $s_1, s_2$ and $x_0$ we
have that $x_1 \neq x_2$.

As $H^{\leq}_{a,a \tr x_0}$ supports $s_1 + \dnorm{ x_0 - s_1
}_{\unitball_1} \unitball_1 $ in $x_0$, it follows that $H^{\leq}_{a,a \tr x_1}$
supports $s_1 + \dnorm{ x_1 - s_1
}_{\unitball_1} \unitball_1 $ in $s_1 + \frac{\dnorm{ x_1 - s_1
}_{\unitball_1}}{\dnorm{ x_0 - s_1
}_{\unitball_1}}(x_0 - s_1) = x_1$. 
Analogously, $H^{\leq}_{a,a \tr x_2}$
supports $s_2 + \dnorm{ x_2 - s_2
}_{\unitball_2} \unitball_2 $ in $x_2$. 
By the same argumentation as before we see that $H^{\leq}_{a,a \tr x_2}$ must
also support $s_1 + \dnorm{ x_2 - s_1 }_{\unitball_1} \unitball_1 =  s_1 + \dnorm{ x_1 - s_1
}_{\unitball_1} \unitball_1$ in $x_2$ (as otherwise we find
a point contradicting \eqref{eq:auxeq}).

Hence, $ s_1 + \dnorm{ x_1 - s_1
}_{\unitball_1} \unitball_1$ is supported in $x_1$ and $x_2$ by the 
halfspaces $H^{\leq}_{a,a \tr x_1}$ and $H^{\leq}_{a,a \tr x_2}$, respectively.
This contradicts the strict convexity of $\unitball_1$.
\end{proof}

In the situation of \cref{thm:lowdimbisectors} we see (as a generalization of
\cite[Lemma 4.1]{BG2012}) that a minimal perturbation of the sites suffices for
\eqref{plp} having a unique optimum.  For the proof we need the notion of
cyclic exchanges. Consider a sequence \linebreak $(i_1, \unit_{j_1}, i_2,
\unit_{j_2}, \ldots, i_r, \unit_{j_r})$ of pairwise distinct cluster indices \linebreak  $i_1,\ldots,i_r
\in \{1,\ldots,k\}$ and pairwise distinct points $x_{j_1}, \ldots, x_{j_r}$.
We define the \emph{cyclic exchange} $\cyclicexchange  :=  \left(   \zeta_{i,j}
\right)_{ \substack{ i = 1,\ldots,k \\ j = 1,\ldots,m }} \in \Rbb^{k \times m}$ 
by
\[\zeta_{i_l,j_l}  := - \frac{1}{\unitweight_{j_l}}, \hspace{3em} \zeta_{i_{l-1},j_l} 
:= \frac{1}{\unitweight_{j_l}}\]
 for $l = 1,\ldots,r$, reading indices modulo $r$, 
and equal $0$ in the remaining components.
 Observe that for any $\clustering_1, \clustering_2 \in \BC$, it follows that
 $\clustering_1 -  \clustering_2$ can be decomposed into
a sum of finitely many scaled cyclic exchanges.

\begin{thm} \label{thm:contuniquesol}
Let $\mX = \Rbb^n$, $\dset := \dfam$ be a $k$-tuple of metrics induced by
strictly convex norms, $\muvector \in \Rbb^k$, and $h: \Rbb \rightarrow \Rbb$
injective.

Then for all $\sitevector  := \Sfam  \in (\Rbb^n)^k$ and
$\epsilon > 0$ there exists an $\tilde \sitevector :=  (\tilde s_1, \ldots,
\tilde s_k) \in (\Rbb^n)^k $ with $ \sum_{i = 1}^k || s_i - \tilde s_i|| <
\epsilon$ such that  for $(\dset,h,\tilde \sitevector)$ the linear program  \eqref{plp} has a unique optimizer.
\end{thm}
\begin{proof}
Suppose that the solution of \eqref{plp}
for $\funcset(\dset,h, S, \muvector)$ is not unique. 

Then there exists $\clustering \in \BC$, a cyclic exchange $\cyclicexchange$
and $\alpha > 0 $ such that $\clustering \pm \alpha \cyclicexchange \in F$,
where $F$ denotes the optimal face of \eqref{plp}.
W.l.o.g. let $\cyclicexchange$ correspond to the sequence
$(1,\unit_1,2,\ldots,r,\unit_r)$.
Since the values of  the objective function of \eqref{plp} are the same it
follows that 
\[\sum_{l = 1}^{r-1}  h(d_{l}(s_{l}, x_{l}))\!-\!
	h(d_{l}(s_{{l}}, x_{{l+1}}))  +  h(d_{r}(s_{r}, x_{r}))\!-\!
	h(d_{r}(s_{{r}}, x_{{1}}))= 0.\]
 Set \[c := \sum_{l = 2}^{r-1}  h(d_{l}(s_{l}, x_{l}))\!-\!
 h(d_{l}(s_{{l}},x_{l+1})) +  h(d_{r}(s_{r}, x_{r}))\!-\!
h(d_{r}(s_{{r}}, x_{{1}})).\] In particular, $c$ does not depend on $s_1$. It
 follows that the set of sites $\tilde s_1$ such that $\cyclicexchange$ is orthogonal to the
 objective function vector of \eqref{plp} is described by the equation
 \[h(d_{1}(x_{1}, \tilde s_{1})) - h(d_{1}(x_{2},
 \tilde s_{1})) = -c.\] 
 With $x_1$ and $x_2$ interpreted as sites, this is the intersection of their corresponding cells of the generalized Voronoi diagram \wrt  $\funcset \big( (d_1,d_1),h, (x_1,x_2), (c,0)\big)$. By \cref{thm:lowdimbisectors}, this set has an empty interior.
Together with the fact that there can only be a finite
 number of cyclic exchanges, the claim follows.
\end{proof}

Finally, similarly to \cite{Carlsson2015}, we may derive the following
continuous version of \cref{cor:optimalclustering} by considering a sequence of
refining discretizations.

\begin{corollary}
Let $\mX = \Rbb^n$, and $\Omega$ a finite continuous measure with $\Omega(\mX) =
\sum_{i = 1}^{k} \kappa_i$, and measurable functions $\hat f_i : \mX
\rightarrow \Rbb$ for $i = 1, \ldots, k$ be given. Further, assume that for $1\leq
i < l \leq k$ and every $c \in \Rbb$ it holds that $\Omega( \{ x \in \Rbb^n: \hat f_i(x)\!-\!\hat f_l(x)\!=\!c\}) = 0$.

Then any partition of $\mX$ into measurable sets $C_i$  with $\Omega(C_i) = \kappa_i$ is optimal \wrt $\sum_{i = 1}^k
\int_{C_i} \hat f_i(x) \Omega(dx)$ if and only if there exists $\mu_1, \ldots,
\mu_k \in \Rbb$ such that with $P_i := \set{x \in \mX: \hat f_i(x) + \mu_i \leq
\hat f_l(x) + \mu_l ~ \forall l}$ it follows that $P_i = \cluster_i$ up to  sets of $\Omega$-measure $0$ for every $i$.
 \end{corollary}

\section{Classes of Generalized Voronoi Diagrams}
\label{sec:classes}

Our general approach can be summarized as follows: We first choose $\dset$
and $h$ depending on the application. 
How this choice is made will depend on which properties of the cells are desired; see \cref{sec:diagrams:euclidean,sec:anisotropicpowerdiagrams,sec:shortestpathdiagrams} for examples; see also \cref{tab:conclusion:thumbrules}.
 Then we make an appropriate, possibly
optimal choice of sites $\sitevector$.   In Euclidean space, for instance,
we can optimize over $\sitevector$ in order to  approximate maximally
consolidated clusterings by evoking results of \cite{BG2010}. Over a discrete
space, we may  heuristically search for sites that minimize the resulting
deviation of cluster weights due to rounding a fractional clustering.  In the
case of anisotropic diagrams, we will use the centers of the current districts
as sites in order to obtain similar new districts.

For any choice of $\sitevector$, we can get a solution $\clustering$
and the feasibility parameter $\muvector$ from  \eqref{plp} and \eqref{dlp}, respectively.
After successive
rounding, we then obtain a clustering together with the feasible
generalized Voronoi diagram \wrt $\funcset(\dset,h, \sitevector, \muvector)$.

\placeFigureC

We will now shortly discuss appropriate choices for $\dset$ and $h$ and
illustrate them by a simple example. In particular, we show how these choices
relate to clusterings with certain characteristics.
 
 \paragraph{An example} \cref{fig:exmpl:instance} shows an instance of
 constrained geometric clustering. Here, the space $\mX$ (gray filled area) is a
 subset of $\Rbb^2$ that is not simply connected.
 The set $X$ consists of $500$
points (blue dots), each of weight $1$. We want to find a clustering of $k=5$
clusters, each of weight $100$. Also, a 
``distance  graph`` 
$\conngraph$  (black lines) is given whose edges are weighted by the Euclidean
distances of their nodes.
For this example, $\conngraph$ was
constructed via a Delaunay triangulation of $\unitset$ and dropping edges
outside $\mX$.  This graph encodes an intrinsic connectivity
structure for $\unitset$. Finally, the black-framed white squares depict
the sites $\sitevector$, which we assume to be pre-determined in this example.
Figures
\ref{fig:exmpl:awvd} to \ref{fig:exmpl:graphclustering} show the clusterings
and supporting diagrams obtained for the  different choices of $\dset$ and $h$
via the methodology described above. 

\subsection{Euclidean Space} \label{sec:diagrams:euclidean}
First, we consider the case that all metrics in $\dset$ are  Euclidean.

\paragraph{Additively Weighted Voronoi Diagrams} An obvious choice is $h
= \operatorname{id}$, \ie, $\dfunc_i(x) := || x - s_i || + \mu_i$ is the Euclidean
distance to the respective cluster's site shifted by $\mu_i$.
Solving \eqref{plp} for a general instance then means to search for a
fractional clustering minimizing the weighted average Euclidean distance to the assigned sites. 
All clusterings in
the relative interior of the optimal face of \eqref{plp} are  supported by
\emph{additively weighted Voronoi diagrams}. For results on the latter see 
\cite{AB1986}, \cite[Chapter 3.1.2]{OBSC2009}, \cite[Chapter 7.4]{AKLK2013}.
Here, Voronoi regions are separated by straight lines or hyperbolic curves and
are in particular star-shaped \wrt their respective sites; see \cref{fig:exmpl:awvd}. If $\mX$ is convex,
this yields connected regions. 
Of course, as the above example shows, this
does not hold for general $\mX$.

 \placeFigureD

\paragraph{Power Diagrams} Taking the squared Euclidean
distances, \ie, $\dfunc_i(x) := || x - s_i ||^2 + \mu_i$,
 yields  \emph{power diagrams} (see \cite{AM1987}, \cite[Chapter
 3.1.4]{OBSC2009}). \cref{fig:exmpl:powerdiagram} shows this case for our
 example. Here, regions are separated by straight lines perpendicular to the
 line spanned by the respective sites. In particular, this yields convex and
 therefore connected regions whenever $\mX$ is a convex subset of $\Rbb^n$.
 Again, connectivity might get lost when $\mX$ is non-convex as is the case in this example. Solving \eqref{plp} may be
interpreted as minimizing the weighted squared error when the clusters are
represented by their respective sites. 

As already pointed out, power diagrams have been thoroughly investigated for
the example of farmland consolidation (\cite{BG2004}, \cite{BBG2014}) and a comprehensive
theory on their algorithmic treatment has been derived (\cite{BG2012}, \cite{BG2010}, \cite{BBG2013}).

Let us close this subsection by briefly recapitulating a result from
\cite{BG2012} that deals with optimal choices of the sites in the strongly balanced case. Recall that a
feasible power diagram is called \emph{centroidal} if 
 \[s_i = \ccenter(\cluster_i) := \frac{1}{\cweight_i} \sum_{j = 1}^m \xi_{i,j}
 \omega_j \unit_j\] for $i = 1,\ldots,k$. The following
  result characterizes centroidal power diagrams as local  maximizers of
  the function $\phi:
  \BC \rightarrow \Rbb$ defined by \[\phi(\clustering) := \sum_{i = 1}^k
  \cweight_i ||\ccenter(\cluster_i)||^2.\]
  Here, some trivial cases of clusters have to be excluded: A clustering $\clustering \in \BC$ is called \emph{proper}, if
 for all $i \neq l$ it holds that
 \begin{equation*}
 	|\suppCi | =  |\suppCl  | = 1  \Rightarrow \suppCi \neq \suppCl.
 \end{equation*}
 
 \begin{thm}[{\cite[Theorem 2.4]{BG2012}}]
  Let $\clustering \in \BC$ be proper.
  Then there exists a centroidal power diagram that supports
  $\clustering$ if and only if $\clustering$ is the unique optimum of
  \eqref{plp} and a local maximizer of  $\phi$.
 \end{thm}
 
Furthermore, finding the global maximum of $\phi$ over $\BC$ is equivalent to
optimizing 
 \begin{align}
 	\min_{\clustering \in \BC} \sum_{i = 1}^k  \sum_{j = 1}^m \xi_{i,j}
 	\unitweight_j || \unit_j - \ccenter(\cluster_i) ||^2.
 	\label{eq:momentofinertia}
 \end{align}
  This may be read as
 minimizing an average variance of clusters, also called the \emph{moment of
 inertia}.
 
 Note that $\phi$ actually depends only on the centers of gravity rather than on the clusters themselves. By
 definition those centers are given by a linear transformation of $\BC$ into the set of all gravity centers in
 $\Rbb^{k d}$. Optimizing $\phi$ over $\BC$ is then equivalent to
 maximizing an ellipsoidal norm in $\Rbb^{k d}$ over the set of gravity
 centers. One can now approximate this norm in the comparably low dimensional
 space $\Rbb^{kd}$ by a polyhedral norm.
 This yields an approximation algorithm by solving a typically manageable number
 of linear programs of type \eqref{plp}; see \cite{BG2010}, \cite{BG2012}. 
 
 Another possibility to derive local
 optima of $\phi$ is by means of a balanced variant of the k-means algorithm
 (see \cite{BBG2013}).
 
\subsection{Anisotropic Diagrams with Ellipsoidal Norms}
\label{sec:anisotropicpowerdiagrams}
Using the Euclidean norm obviously cannot pay regard to the shape of $\mX$ nor
any other information about the desired extents or orientations of the
resulting clusters. One possibility of including such information is to use
 anisotropic diagrams.

While we could, in principle, employ arbitrary norms we will consider $\dset$
to be induced by ellipsoidal norms in the following.
So, let $M_i \in \Rbb^{n
\times n}$  be symmetric positive definite matrices defining the ellipsoidal
norms via \[\anormi{x} := \sqrt{ x\tr M_i x }, \]  
$i = 1,\ldots,k$.  In our main application, the matrices are chosen so as to obtain clusters
similar to a pre-existing structure (cf. \cref{sec:res:anisotropicpowerdiagrams}).

As in the Euclidean case, we consider the transformation functions $h =
\operatorname{id}$ and $h = (\cdot)^2$.
For  $h = \operatorname{id}$ we obtain \emph{anisotropic Voronoi diagrams}.
 These have already been applied in the field of mesh generation
 (\cite{LS2003}, \cite{CG2011}) on a Riemannian manifold in $\Rbb^n$ provided
 with a metric tensor $M: \mX \rightarrow M^{n \times n}$. Hence, the functions
 $\dfunc_i$ can be seen 
as local approximations of the geodesic distances \wrt that manifold.  Even
without additive weights, the resulting diagrams need not be connected.

We will refer to case $h = (\cdot)^2$ as \emph{anisotropic power diagrams}. In
\cite{ABG+2015} these  were
 successfully used for the  reconstruction of polycrystalline structures, where
 information about volumes as well as moments of crystals is given. Regions are
 separated by straight lines or quadratic curves.
 Figures \ref{fig:exmpl:anisotropvd} and \ref{fig:exmpl:anisotroppd} show the
 case of additively weighted anisotropic Voronoi diagrams and anisotropic power
 diagrams, respectively. The dotted ellipses depict the unit balls of the
 respective norms.
 
 \placeFigureE

\subsection{Shortest-Path Diagrams} \label{sec:shortestpathdiagrams}

Even in the anisotropic case, the diagrams considered so far might fail in
depicting intrinsic relations of the points in $\mX$. In our application of
electoral district design, this occurs as points are only representatives of
their  municipalities' regions. Thus, information about neighborhood
relations is  lost (\cf \cref{sec:results}). In such
cases, generalized Voronoi diagrams on an edge-weighted graph $\conngraph = (\unitset, \arcset, \arcdistfunc)$ can be preferable.

\cref{fig:exmpl:graphclustering} shows the result if $\mX$ is the discrete space
of all elements in $\unitset$  together with the metric induced by  $\conngraph$ (see \cref{sec:introduction}).
Taking $ \dfunc_i(x) := \dgraph(s_i,x) + \mu_i$, this means that the weighted average
distance of  elements to their assigned sites in the graph is minimized. We will refer to
this case as \emph{shortest-path diagrams}. Of course, if $\conngraph$ is
complete and edge weights are the
Euclidean distances between the vertices, this again results in additively
weighted Voronoi diagrams.
  
  \placeFigureF

In general, there are two main motivations to use a discrete space $(X,
\dgraph)$.
The obvious first reason are applications that live on some sort of graph. 
For instance, \cite{OSF+2008}
 proposes to use Voronoi diagrams on discrete spaces for the
 representation of service areas in cities
 as the Euclidean norm is often unsuitable for approximating distances in urban
 traffic networks.
Second, there are applications that require clusters to be connected in some
sense.
 Often, of course, this connectivity is already established when
 clusters are supported by a diagram in $\mX$ with connected regions. However,
 as has been observed before, this is not always the case with the diagrams
 proposed so far, in particular, since the underlying space might not be convex. 

We say that a fractional clustering $\clustering \in \BC$ is
\emph{$\sitevector$-star-shaped} for sites $\sitevector = \Sfam \in \unitset^k$, if for every $i \in \{1,\ldots,k\}$ and
$\unit \in \suppCi$ it follows that $v \in \suppCi$ for every $v$
on a shortest $s_i$-$\unit$-path in $\conngraph$.
By the following result, clusters that are supported by shortest-path diagrams
are $\sitevector$-star-shaped. More precisely, the Voronoi regions are shortest-path
trees rooted at their  sites. 

\begin{thm} \label{thm:starshapedgraph}
Let  $(X,\dgraph)$ be the metric space induced by a 
connected and edge-weighted graph $\conngraph = (\unitset, \arcset,
\arcdistfunc)$.
Let $\clustering \in \BC$, $\sitevector \in \unitset^k$, $\muvector \in
\Rbb^k$, and define  $\dset := \dfam$ via $d_i = \dgraph$ for $i= 1,\ldots,k$. 
 
If the generalized Voronoi diagram $\voronoidg$ \wrt
$\funcset(\dset,\id,\sitevector,\muvector)$ supports $\clustering$, then
$\clustering$ is $\sitevector$-star-shaped. In particular, $s_i \in \suppCi$
holds for all  $i \in \{1,\ldots,k\}$.
\end{thm}

\begin{proof}
Let $F$ be the optimal face of \eqref{plp}, then $\clustering \in \relint(F)$ holds
by \cref{thm:supportingclustering}.
For some $i \in \{1,\ldots,k\}$, let $\unit_j \in \suppCi$ and let  
 $s_i = v_1, v_2,\ldots,v_t := \unit_j$ be a shortest $s_i$-$\unit_j$-path in
 $\conngraph$.

Suppose that there exists $x_p$  such that  $v_l = \unit_p \not \in \suppCi$ 
 for some $l \in \{1,\ldots,t-1\}$. By the feasibility of $\clustering$
there exists $r \in \{1,\ldots,k\}$ such that $\unit_p \in \supp(\cluster_r)$.

Due to the choice of $\unit_p$ and $\unit_j$ it follows that $\xi_{i,p} = 0 <
1$, $\xi_{i,j} > 0$, $\xi_{r,j} \leq 1- \xi_{i,j} < 1$ and $\xi_{r,p} > 0$.

Let $\cyclicexchange \in \Rbb^{k \times m}$ be the cyclic exchange for the
sequence $(i,\unit_j,r,\unit_p)$ (as defined in
\cref{sec:method:clusteringdiagramcorrespondence}).
Then there exists some $\alpha > 0$ such that $\clustering
+ \alpha \mathcal{Z} \in \BC$. If $\alpha$ is sufficiently small $\clustering
+ \alpha \mathcal{Z}$ has (at least) one non-zero component more than
$\clustering$. Since $\clustering \in \relint(F)$, it follows that $\clustering
+ \alpha \mathcal{Z} \not \in F$. Thus, 
$0 <  
\dgraph(s_i,\unit_{p}) - 
\dgraph(s_i,\unit_{j})  +  \dgraph(s_r,\unit_{j}) - \dgraph(s_r,\unit_{p})$.
 
 Now, by the triangle inequality, $ \dgraph(s_r,
 \unit_j) \leq \dgraph(s_r, \unit_p) + \dgraph(\unit_p,\unit_j)$.  As $\unit_p$ lies on a
 shortest $s_i$-$\unit_j$-path, it furthermore holds that $\dgraph(s_i,\unit_j) = \dgraph(s_i,\unit_p) +  \dgraph(\unit_p, \unit_j)$. 
 
 Together, this yields
$\dgraph(s_r,\unit_j) - \dgraph(s_i,\unit_j) \leq  \dgraph(s_r, \unit_p) - \dgraph(s_i,\unit_p)$,
 a contradiction.
 
  As $\suppCi \neq \emptyset$ for
  $i = 1,\ldots,k$, this in particular implies $s_i \in \suppCi$.
\end{proof}

\placeFigureG

 The requirement that $\clustering \in \BC$ lies in the relative interior of
 the optimal face \eqref{plp} is crucial for shortest-path distances to
 preserve star-shapedness. In  \cite{Zoltners1983a}, a Lagrange relaxation
model of an integer version of \eqref{plp} for shortest-path distances was
considered and connectivity of resulting clusters was demanded,
while it was pointed out in \cite{Schroder2001}, that this may not hold
whenever non-unique optima appear. \cref{thm:starshapedgraph} clarifies 
the situation: in view of \cite{Schroder2001} it is precisely the requirement that
the solution lies in the relative interior of the optimal set.

Let us now consider the  example of  \cref{fig:proof:connectedonlyaffine} with
$a:=1, b:= 1$, $c := 2$ and the resulting constrained clustering instance
for $\uweightvector := \ones$, $\cweightvector := (2,2)\tr$.
We obtain the two clusterings $\clustering^{(a)}$,  $\clustering^{(b)} \in \BC$
 via   
 \[\cluster^{(a)}_1 := (1,\frac{1}{2},\frac{1}{2},0), \hspace{2em} \cluster^{(a)}_2 :=
 (0,\frac{1}{2},\frac{1}{2},1)\]
  and 
  \[\cluster^{(b)}_1 := (1,0,1,0), \hspace{2em}
  \cluster^{(b)}_2 := (0,1,0,1).\]
   The generalized Voronoi diagram
$\voronoidg^\ast$ \wrt $f_1(x) = \dgraph(s_1,x) +1$ and  $f_2(x) = \dgraph(s_2,x) $ (\ie, $h = \id$ and $\muvector = (1,0)\tr$), consists of the cells
$\voronoi^\ast_1 = \{ x_1, x_2, x_3 \}$ and $\voronoi^\ast_2 = \{ x_2, x_3, x_4
\}$. Thus, it is feasible for both
$\clustering^{(a)}$ and $\clustering^{(b)}$. Hence, they are both minimizers of
\eqref{plp} by \cref{thm:complementaryslackness}.
However, only $\clustering^{(a)}$ is supported by $\voronoidg^\ast$ and
$\sitevector$-star-shaped while $\clustering^{(b)}$ contains a disconnected
cluster.

More generally, this happens whenever shortest paths intersect. 
This can have a dramatic effect on integer assignments. In fact, in
order to conserve connectivity after rounding, whole fractionally assigned branches of the
shortest-path trees might have to be assigned to one of their supporting
clusters. Of course, this results in greater deviations of cluster weights.
For our running example, \cref{fig:exmpl:graphclustering:fracts} depicts the
points that have originally been fractionally assigned to both the blue and the green
colored cluster and are now fully assigned to the green
cluster in the integer solution. 

\placeFigureH

A natural idea to overcome this effect as well as to obtain more consolidated
clusters is to try to imitate the idea of squaring the  distances (that led to power diagrams in the Euclidean space) to the discrete space
$(X,\dgraph)$. 

Let us once more consider the previous  example from just above.
This time, however, we take the generalized Voronoi diagram $\voronoidg'$
\wrt $f_1(x) = \dgraph(s_1,x)^2 +4$ and $f_2(x) =
\dgraph(s_2,x)^2$ (\ie, $h = (\cdot)^2$ and $\muvector = (4,0)\tr$). It
follows that $\voronoi'_1 = \{ x_1, x_3 \}$ and $\voronoi'_2 = \{ x_2, x_4 \}$.
Thus, $\voronoidg'$ supports $\clustering^{(b)}$ but is not even feasible for
$\clustering^{(a)}$. In fact, $\clustering^{(b)}$ is the unique minimizer of
\eqref{plp} for this choice of $h$ and does not yield connected clusters.
\cref{fig:exmpl:graphclusteringSquared:nonconnected} demonstrates the same
effect for our running example. Here, the single yellow unit in the center
is separated from its respective cluster.

Unfortunately, the following theorem shows that this is a general effect. 
In fact, for any transformation function $h$ that is
 not affine, clusters can be disconnected despite being supported by a
 corresponding diagram. Thus, if connectivity is to be guaranteed a-priorily,
 this dictates the choice of shortest-path diagrams in our approach.
 
 \begin{thm} \label{lem:connectedonlyaffine}
Let $(X,\dgraph)$ be  the  metric space induced by a 
 connected and edge-weighted graph $\conngraph$. Let $\clustering \in \BC$,
 $\sitevector \in \unitset^k$, and $\muvector \in \Rbb^k$, and 
define  $\dset = \dfam$ via $d_i = \dgraph$ for $i= 1,\ldots,k$. 
  Furthermore, let $h: \Rbbp \rightarrow \Rbb$  be continuous.

  If $h(x) =  \alpha \cdot
 x + \beta$ for some $\alpha \in \Rbbp, \beta \in \Rbb$
  and the generalized Voronoi diagram \wrt $\funcset(\dset,h,\sitevector,\muvector)$ supports
  $\clustering$, then $\clustering$ is $\sitevector$-star-shaped.
 
  If $h$ is any continuous function  but not of  the above type, this
  is not  true in general.
 \end{thm}
 
\begin{proof}
The first claim  follows by  replacing $\dgraph$ with $\alpha \cdot \dgraph +
\beta$  in the proof of \cref{thm:starshapedgraph}. Note that in the case
$\alpha = 0$   the whole set $\BC$ is optimal, which causes all components of a
solution from the relative interior to be strictly positive.

For the second claim, let some continuous function $h : \Rbbp \rightarrow \Rbb$
be given. Consider the graph from \cref{fig:proof:connectedonlyaffine} with
$\unitset = \{\unit_1,\unit_2,\unit_3,\unit_4 \}$, edges
 $\{\{\unit_1,\unit_2\}$, $\{\unit_2,\unit_3\}$, $\{\unit_2,\unit_4\}\}$ and
 edge weights  $\arcdistfunc(\{\unit_1,\unit_2\} ) = a$,
 $\arcdistfunc(\{\unit_2,\unit_3\} ) = b
$ and  $\arcdistfunc(\{\unit_2,\unit_4\} ) = c $ for some $a,b,c \in \Rbbsp$.
Furthermore, let $\unitweight_1 = \unitweight_2 = \unitweight_4 = 1$,
$\unitweight_3 = 3$ and $\cweight_1 =  \cweight_2 = 3$.
By \cref{thm:supportingclustering} and since $\BC \neq \emptyset$ there exists a
clustering $\clustering \in \BC$ that is supported by the generalized Voronoi
diagram $\voronoidg$ \wrt $\funcset(\dset,h,\sitevector,\muvector)$ for some
$\muvector  \in \Rbb^2$.

Now suppose that $\clustering$ is $\sitevector$-star-shaped.
Together with the
choice of the weights this implies 
$\{ x_2, x_3 \} \subset \supp(\cluster_1)
\cap \supp(\cluster_2) \subset \voronoi_1 \cap \voronoi_2$.
 Hence, one gets
$h( \dgraph(s_1,x_j)) + \mu_1 = h(\dgraph(s_2,x_j)) + \mu_2$ for $j = 2,3$. 
Subtraction of
the equality for $j = 3$ from the one for $j = 2$ and insertion of the
shortest paths lengths yields
\begin{align*}
 	h( a ) - h ( a + b) = h( c ) - h( c + b) .
 \end{align*}
 Setting $\tilde h := h - h(0)$, taking the limit $c \rightarrow 0$ and
 using the continuity of $h$ this implies
\[\tilde h (a + b ) = \tilde h (a  )  +\tilde h (b).\]
Since $a,b \in \Rbbsp$ are arbitrary and $\tilde h$ is continuous, it 
readily follows that $\tilde h$ is linear on $\Rbbp$ and thus $h$ is of the form
$h(t) = \alpha \cdot t + \beta$.

In order to see that $\alpha \geq 0$, it is sufficient to consider the example
$\unitset = \{ \unit_1, \unit_2 \} $,  $s_i = x_i, 
\unitweight_i = 1, \cweight_i = 1$ for $i = 1,2$ and a single edge $\{x_1,x_2\}$
of arbitrary positive length. If $\alpha < 0$, then the optimal clustering is
$\supp(\cluster_1) = \{ x_2 \}$ and $\supp(\cluster_2) = \{ x_1 \}$, which
contradicts the claim of \cref{thm:starshapedgraph}.
 \end{proof}

\section{Application to the Electoral District Design Problem}
\label{sec:results}
We will now apply our approach  to the electoral district design problem. We
 show the effect of using power diagrams, anisotropic power diagrams and
 shortest-path diagrams for the example of the Federal Republic of Germany.
 
\subsection{Electoral District Design}
\label{sec:method:electoraldistrictdesign}

 Despite the differences in voting systems, the issue
of designing electoral districts can be stated in a common way suitable for most
of them. A state consists of  basic units such as municipalities or smaller
juridical entities. Units are of different weight, usually and in the following their number of eligible voters. Units are then to be partitioned into a given number of
districts of (approximately) equal weight. Hence, we are facing a constrained clustering problem
as defined in \cref{sec:method:constraintedclustering}. 

Usually, districts are required to be
``compact" and ``contiguous'' (cf. \cite{RSS2013}). Both are, not formally
defined juridical demands requiring a proper mathematical modelling. 
How to define a measure for ``compactness'' of districts has, in fact, been
widely discussed in the literature (\eg  \cite{You1988}, \cite{NGCH1990},
\cite{HHV1993}, \cite{Alt1998}).
One widely accepted measure (\cite{HWS+1965}, \cite{FH2011}) is the moment of inertia
as defined by \eqref{eq:momentofinertia}, where each municipality is
modelled as a point in the Euclidean plane.

Contiguity usually requires that the area belonging to a district's
municipalities is connected.
This can be modelled by the adjacency graph $\conngraph$ with nodes
corresponding to the units and edges between two units whenever they share a common 
border (which can be crossed).
Connectivity of clusters is then defined by demanding that each induced subgraph
$\conngraph[\suppCi]$ is connected. 
The edges of $\conngraph$ can  be weighted, for
example,  by driving distances between the corresponding municipalities'
centers.

Due to census development electoral districts have to be regularly adapted.
Therefore, a further requirement may be to design districts that are similar
to the existing ones. Let $\clusteringO$  be the integer clustering that
 corresponds to the original districts and  $\clustering^\ast$
 be a newly created integer clustering.
 One  may measure the difference of the two district plans by the ratio of voter
 pairs that used to share a common district but are now assigned to different
 ones. With $A(\clusteringO,\clustering^\ast) := \{ (j,r) : 1 \leq j < r
 \leq m \wedge \exists i :
 \xiO_{i,j} = \xiO_{i,r} = 1 ~ \wedge \forall i : \xi^\ast_{i,j} \cdot
 \xi^\ast_{i,r} = 0  \}$ this is, more formally, given by
\begin{align} \label{eq:formula_changedpairs}
 	\frac{1}{ \sum_{i = 1}^k  \tbinom{ \omega(\clusterO_i) }{2}  }
 	\sum_{(j,r) \in A(\clusteringO,\clustering^\ast)} \omega_j \cdot \omega_r.
 \end{align}

\subsection{Dataset Description} \label{sec:results:data}

By the German electoral law \cite{bwg}, there are  a
total of 299 electoral districts that are apportioned to the federal states according to
their population. As districts do not cross state borders, each
state must be considered independently. A district's size is measured by its number of
eligible voters (\cite{bvg2012}). It should not deviate from the federal average by more
than 15\%; a deviation exceeding 25\% enforces a redesign of districts. 
As far as possible, for administrative and technical reasons municipal borders
must be conserved. Furthermore, districts are
supposed to yield connected areas.

For our application the data from the last German election held on September
22nd 2013 was used. Statistics about the number of eligible voters were taken
from \cite{WdBudLaen2014}. Geographic data for the municipalities and election
districts of 2013 were taken from \cite{GVG250} and \cite{BTW13G}, respectively. 
Greater cities which on their own constituted one or
 several election districts in 2013  were  neglected for our
 computations as any proper district redesign would imply to split these
 municipalities up into smaller units and thus, of course, required data on a
 more granular basis. For the same reason, the city states (Berlin, Bremen,
 Hamburg) were not taken into consideration.

 \cref{fig:res:dev:2013} depicts the deviation of clusters sizes of the 2013
 election from the federal average.  Accordingly, in the 2013 elections 57 of the 249
 districts that were considered  had a deviation of more than 15\%.
  The overall average deviation is 9.5\%. The maximum deviation of 25.15\% is
  obtained for a district in the state of Bavaria.

 We have identified each municipality by a point in the plane given by its geographic
 coordinates in the EPSG 3044 spatial reference system (\cite{Wik2015}).
 For the shortest-path clustering approach, we have an edge in the graph
 $\conngraph$ between two municipalities whenever their regions share a 
 border. The edge lengths were taken as  driving distances obtained from
 the MapQuest Open Geocoding API (\cite{OMQ}).
 
In the following of this chapter, we state the practical results for all of
Germany and for various single states.  The latter are typical examples, \ie,
the individual results for the other states are very similar;  see
\cref{tab:res:overview} and  \url{http://www-m9.ma.tum.de/material/districting/}.
 
 \subsection{Power Diagrams}

As pointed out in \cref{sec:diagrams:euclidean}, clusterings with minimal
moment of inertia are supported by centroidal power diagrams. Thus, power
diagrams yield highly consolidated district plans.

 Squared Euclidean distances were already used  for the problem of electoral
 district design; see e.g. \cite{HWS+1965} and \cite{Hoj1996}. Centroidal power
 diagrams have been used by \cite{FH2011}, who presented a gradient descent  
 procedure similar in spirit to the balanced $k$-means approach (\cite{BBG2013}).

 In our approach, first a fractional clustering
that is supported by a centroidal power diagram was created.
Here, the sites were determined as approximations of the global optimizers of 
\eqref{eq:momentofinertia} as proposed in \cite{BG2010}.
Second, the fractionally assigned units were rounded optimally with respect to
the resulting balancing error.

As it turns out, non-connected districts do indeed sometimes occur. This is due
to the non-convexity of the states in general and particularly  due to ``holes''
in the state areas resulting from excluded municipalities or city states.
In many cases, this may not regarded a problem. However, since we insisted on connected districts we applied some post-processing.
After running the approach as described in \cref{sec:diagrams:euclidean}, the resulting districts were checked for connectivity. This was done in an automated manner using the adjacency graph of neighboring units and standard graph algorithms. If unconnected districts were detected, the program (P) was rerun under preventing or forcing municipalities to be assigned to a certain district by constraining the corresponding decision variables to 0 or 1, respectively. For example, if a district had been split into two parts by a geographical barrier such as a lake or an indentation in the state border, the smaller part (which mostly consisted of a a single municipality) was excluded from being assigned to that district. This was comfortably done using a graphical user interface and required only a few, if any, iterations per state. For the considered data,
a total of 51 (0.46\%) municipalities was preassigned in order to establish
connectivity.
 
 \cref{fig:res:bayern} shows the
 original districts of the state of Bavaria from the 2013 elections compared to
 the results of the power diagram approach.
 \cref{fig:res:bayern:cpd} furthermore depicts the resulting polyhedral power
 diagram regions. Here, three units had to be fixed in order to establish connectivity.

\placeFigureI

  \subsection{Anisotropic Power Diagrams}
  \label{sec:res:anisotropicpowerdiagrams}
 
 Next, we show how to use anisotropic power diagrams in order to
 obtain clusters that are similar to those of the 2013 election.
 
 As in \cite{ABG+2015}, a principal component analysis was performed in
 order to determine a local ellipsoidal norm for each district. Let 
  $\clusteringO := \left(   \xiO_{i,j} \right)_{ \substack{
 i = 1,\ldots,k \\ j = 1,\ldots,m }}$ be the integer clustering encoding the
 original districts of some state. For each $i =1,\ldots,k$, the
 covariance matrix  $V_i$ is  computed as 
 \begin{align*}
 	V_i := \sum_{j = 1}^m
 \frac{ \xiO_{i,j} \unitweight_j}{ \unitweight(\clusterO_i)} \left( \unit_j -
 \ccenter(\clusterO_i) \right)  \left( \unit_j - \ccenter(\clusterO_i) \right)\tr. 
 \end{align*}
  Using a singular
 value decomposition, we obtain an orthogonal matrix $Q \in \operatorname{O}(2)$
 and $\singv_1^{(i)}, \singv_2^{(i)} > 0$ such that 
 
\[V_i = Q \begin{pmatrix}
\singv_1^{(i)} & 0 \\ 0 & \singv_2^{(i)}
\end{pmatrix} Q\tr.\] 
We then  set 
\begin{align*}
 	M_i :=  Q \begin{pmatrix}
(\singv_1^{(i)})^{-1} & 0 \\ 0 & (\singv_2^{(i)})^{-1}
\end{pmatrix} Q\tr 
 \end{align*}
 in order to obtain an  ellipsoidal norm as described in
\cref{sec:diagrams:euclidean}.
With the centroids of $\clusteringO$ as starting sites we performed a
balanced k-means algorithm (see \cite{BBG2013}) in order to obtain centroidal
anisotropic power diagrams.

As in the case of power diagrams and due to the same reasons, again not all of
the resulting clusters were connected. Applying the same post-processing this
could again be treated, affecting a total of 33 (0.30\%) municipalities.

  \cref{fig:res:niedersachsen}
 shows the 2013 elections' districts of the state of Lower Saxony and the
 results of the anisotropic power diagram approach.
   The blue ellipses in \ref{fig:res:niedersachsen:apd} depict the unit balls of
   corresponding local cluster norms. For this state no post-processing was necessary.

\placeFigureJ

 \subsection{Shortest-Path Diagrams} \label{sec:res:shortestpathdiagrams}

 In order to enforce connectivity directly,  we apply the shortest-path diagrams
 of \cref{sec:shortestpathdiagrams} \wrt the adjacency
 graph $\conngraph$ of neighboring units.
 Shortest-path distances have appeared in the context of electoral district
 design several times (\eg, \cite{Segal1977}, \cite{Zoltners1983a},
\cite{Schroder2001}, \cite{KNS2005}, \cite{RS2008}). In \cite{RSS2008}
also  Voronoi diagrams on a connectivity graph were considered but
multiplicative rather than additive weights were evoked, which led to substantially 
bigger balancing errors.
In \cite{Zoltners1983a} and \cite{Schroder2001}, Lagrangian relaxations were applied
which are, naturally, closely related to our methodology.

 In our approach,  the effect of non-unique optima and therefore more
 fractionality could be
 observed as predicted in \cref{sec:shortestpathdiagrams}. This was again
 handled in a post-processing phase by the implementation of a slightly more sophisticated rounding procedure. Here, fractional components
 were rounded in an iterative manner while updating the shortest-path distances
 to the already integrally assigned clusters in each step. In order to determine
 suitable sites $s_i$, $i = 1,\ldots,k$, a simple local search \wrt  improvement
 of the resulting deviation of cluster sizes was performed. Here, the units closest to the
 centroids of $\clusteringO$ served as starting sites.
 
  \cref{fig:res:nrw} shows the
   original districts of the state of North Rhine-Westphalia from the 2013
   elections compared to the results via shortest-path diagrams.
   The edges in \cref{fig:res:nrw:sp} furthermore depict the shortest-path trees
   connecting the resulting clusters. 
 
 \placeFigureK

  \subsection{Evaluation}

We will now summarize the results of our different \linebreak approaches. 
See \cref{tab:res:overview} in the appendix for a more detailed overview of the
results for the different German states.
  
\placeFigureL
  
As already pointed out,  all approaches led to
district plans complying with the German electoral law, \ie
obeying the deviation limits and connectivity of districts. 

\begin{table}[h!tbp]
\centering
	\begin{tabular}{r|C{1.4cm}@{\hspace{1mm}}C{1.4cm}@{\hspace{1mm}}C{2.2cm}@{\hspace{1mm}}C{1.8cm}}
	 &  \small{Districts 2013} & \small Power Diagrams &\small Anisotropic Power Diagrams & 
	 \small  Shortest-Path Diagrams
	 \\
	  \hline
	 Avg. &   9.52\% &   2.69\% &   2.73\% &    \textbf{2.13}\% \\
	 Max. &  25.1\% &  10.60\% &  14.71\% &   \textbf{9.73}\% \\
    \end{tabular}
    \caption{Deviations of district sizes from the federal
    average size.} \label{tab:res:deviations}
\end{table}

The largest maximal deviations  occur for the states of \linebreak Mecklenburg-Vorpommern (14.71\%) and North  Rhine-\linebreak Westphalia (14.34\%), both for the anisotropic power diagram approach.  
However, \cref{tab:res:overview} shows that even those deviations are not far off from optimality.
In fact, the average district size in Mecklenburg-Vorpommern itself is already 8.9\% below the federal average. In North Rhine-Westphalia, the high population density results in units of greater weight and thus greater rounding errors. As expected, \cref{tab:res:overview} shows that states with a finer division into municipalities generally yield smaller deviations.

Figure \ref{fig:res:deviations} depicts the deviations of district sizes for
our approaches and \cref{tab:res:deviations} contains the average and maximal
values for all our approaches and the elections of 2013.
While all approaches perform well, the shortest-path diagram approach
  is slightly superior.  This is not surprising, as the local search in the
  shortest-path approach only focuses on the deviation error.
  
  \begin{table}[h!tbp]
\centering
	\begin{tabular}{r|C{2.3cm}@{\hspace{1mm}}C{2.3cm}@{\hspace{1mm}}C{2.3cm}}
	  & \small Power Diagrams &  \small Anisotropic Power Diagrams &  \small 
	  Shortest-Path Diagrams
	 \\
	  \hline
	$\Delta$MoI &  \textbf{-11.3\%} &   1.3\% &   -0.5\%
    \end{tabular}
    \caption{Relative change of the moment of inertia
    as defined by \eqref{eq:momentofinertia} compared to 2013.} 
    \label{tab:res:moi}
\end{table}
  
 \cref{tab:res:moi} yields the relative change of the total moment of  inertia (as
 defined in \eqref{eq:momentofinertia})
 compared to 2013. According to this measure, power diagrams lead to the by far
 best consolidation.
 Shortest-path diagrams yield slightly more and the anisotropic power-diagram slightly less
 consolidated districts. However, recall that the moment of inertia is measured
 in Euclidean space, while the anisotropic power diagrams minimize local
 ellipsoidal norms. Hence, a fair comparison should really involve a measure
 similar to \eqref{eq:momentofinertia} but based on  those local 
 norms.

    \begin{table}[h!tbp]
\centering
		\begin{tabular}{r|C{2.2cm}@{\hspace{1mm}}C{2.3cm}@{\hspace{1mm}}C{2.2cm}}
		& \small Power Diagrams &  \small Anisotropic Power Diagrams &  \small 
		Shortest-Path Diagrams
		\\
		\hline
	$\Delta$Pairs & 40.6\% &  \textbf{21.4\%} &  35.4\% \\
    \end{tabular}
    \caption{ Total ratio of changed pairs of voters
    as defined by \eqref{eq:formula_changedpairs}. } 
    \label{tab:res:changedpairs}
\end{table}
  
  In order to compare the obtained districts to the ones of 2013, the 
  ratio of changed pairs according to \eqref{eq:formula_changedpairs} over the
  districts of all states are shown in \cref{tab:res:changedpairs}. Here, indeed
  the \linebreak anisotropic power diagram approach separates significantly less pairs of voters that used
	  to vote in the same district in the 2013 election.
  
  \placeFigureM

\cref{fig:res:hessen}  shows the results of the state of Hesse for all
approaches in direct comparison. In particular, they illustrate the numbers
listed above.  The power diagram districts seem most consolidated, while
elongated districts appear in the shortest-path result. Also, a higher degree of
similarity of the districts from anisotropic power diagrams to the original
districts can be recognized.

Finally, concerning the computational running times of our approaches, note that
once the parameters $(\dset,h,\sitevector)$ are determined, a simple linear program
\eqref{plp} in dimension $k \times m$ with $k + m$ constraints and $2km$ non-zero entries is to be solved.
This, of course, is unproblematic even for fairly large instances (such as, 
for example, $10^5$ municipalities and $10^3$ districts) using state-of-the-art solvers.

When, however, the structural parameters are part of the optimization process, 
the computational scalability strongly depends on the chosen approach: In case of
power diagrams, an approximate optimization of \eqref{eq:momentofinertia} (\cf
\cref{sec:diagrams:euclidean}) also leads to solving a number of
linear programs in dimension $k \times d$ and a fairly small number of
constraints. However, this means approximately maximizing an ellipsoidal norm, 
which is an $\NP$-hard problem with no constant-factor approximation 
unless $\mathcal{P} = \NP$ (cf. \cite{BGKKLS98}, \cite{BGKKLS01}, \cite{Bri2002}, \cite{BG2012}). 
Thus, this will be problematic for huge $k$. However, as the number of representatives 
is limited and, particularly, voting is in effect often restricted to subpopulations (like the
states within the Federal Republic of Germany), this remains tractable in practice (as demonstrated).

In the case of anisotropic power diagrams, the applied balanced
k-means variant reduces to solving \eqref{plp} a few times.

As we applied a local-search of sites for shortest-path diagrams, there, the
running times are, of course, highly dependent on the size of the considered
neighborhoods that are evaluated. In our computations, we restricted a sites
vector's neighborhood to single site-exchanges with the respective 50 closest
units. Then, if the local search is performed separately for each cluster,
we can again expect good scalability in terms of $k$.

 For our data sets, the computations ran on a standard PC within seconds for  anisotropic 
 power diagrams, within few hours for (approximately) centroidal power diagrams and 
 were in the range of several hours for the shortest-path approach. In any case, for 
 our application of electoral districting the computation times were not critical.

\section{Conclusion}
\label{sec:conclusion}
We  presented a unifying approach to constrained geometric clustering in
arbitrary metric spaces and applied three specifications to the electoral
district design problem.
We used a relation between constrained fractional clusterings and
additively weighted generalized Voronoi diagrams which is based on LP-duality.
In particular, we obtained clusterings with prescribed cluster sizes
that are embeddable into additively weighted generalized Voronoi diagrams. A short
discussion of typical classes of diagrams as well as details for the special
cases of power diagrams and shortest-path diagrams on graphs were provided.

  \begin{table}
	\begin{tabular}{L{2.2cm}|C{1.4cm}@{\hspace{1mm}}|@{\hspace{1mm}}C{2.2cm}@{\hspace{1mm}}|@{\hspace{1mm}}C{1.8cm}}
		& \small Power Diagrams 
		& \small Anisotropic Power Diagrams 
		& \small  Shortest-Path Diagrams
		\\
		\hline
		\small consolidation   &  $\oplus$ $\oplus$ &  $\oplus$  &   \\ \hline
		\small  connectivity &   $\oplus$ &  &   $\oplus$   $\oplus$  \\ \hline
		\small  conservation of existing structure & & $\oplus$   $\oplus$  & $\oplus$  \\
	\end{tabular}
	\caption{"rule of thumb" for the choice of diagram types} 
	\label{tab:conclusion:thumbrules}
\end{table}

Results for the example of electoral district design in Germany prove to be very
favorable with respect to the deviations from prescribed cluster sizes of the
obtained
integer clusterings.
In particular, we  pointed out how the choice of a
class of diagrams can  be made for the different optimization criteria.
\cref{tab:conclusion:thumbrules} summarizes our observations by providing an
informal "rule of thumb" for the choice of a diagram type:
If  particularly consolidated districts are desired,
power diagrams yield the best results. As they further produce convex cells, the
resulting districts are likely to be connected whenever the units can be
approximated well by points in the plane and the state area is close to convex. 
If districts are preferred that are 
similar to existing structures, anisotropic power diagrams perform very well.
Due to their relation to power diagrams, they have favorable characteristics
\wrt consolidation   as well. Connectivity is  guaranteed a-priorily by
shortest-path diagrams. Note that with edge
distances obtained from  anisotropic norms, conservation of existing structures
may here be achieved, too.
Thus, our methodology is capable of satisfying different requirements that may
occur for political reasons.


\printbibliography

 \begin{appendices}
 \crefalias{section}{appsec}

  \begin{table*}[!ht]
  	\section{Results for the German federal election}
  	\centering
  	\resizebox{!}{.95\textheight}{
  	\rotatebox{90}{
  \begin{minipage}{1.05\textheight}
  	\begin{tabular}{L{2.3cm}|R{.75cm}|R{.5cm}|R{1.2cm}||c|rrr|rrr|rrr}
	    \toprule
		 	&\multicolumn{3}{c||}{}
	       & \multicolumn{1}{c|}{\textbf{2013}}
	       &\multicolumn{3}{C{3.2cm}|}{\textbf{Power Diagrams}}
	       &\multicolumn{3}{C{3.2cm}|}{\textbf{Anisotropic Power Diagrams}}
	       &\multicolumn{3}{C{3cm}}{\textbf{Shortest-Path Diagrams}}\\
	    \textbf{State} 
	    	&\multicolumn{1}{c|}{$m$}
	    	& \multicolumn{1}{c|}{$k$} 
	    	& \multicolumn{1}{c||}{State Dev} 
			& \o Dev
	    	& \o Dev  & $\Delta$MoI  &  $\Delta$Pairs
	    	& \o Dev  & $\Delta$MoI  &  $\Delta$Pairs
	    	& \o Dev  & $\Delta$MoI  &  $\Delta$Pairs  \\
	    	&  &  &   &   in \% &  in \% & in \% &  in \% &  in \% &  in \% & in \% & in \% & in \% & in \% \\
	    \midrule
		Baden-W\"urttemberg & 1100 & 36 &   1.9 &   7.9 &   1.9 & -13.1 &  41.0 &   2.1 &   0.6 &  21.3 &   1.9 &   0.9 &  35.2 \\
		Bavaria & 2055 & 40 &   0.9 &  13.8 &   1.4 & -14.8 &  44.5 &   1.2 &   4.3 &  37.0 &   0.9 &   2.7 &  44.1 \\
		Brandenburg & 419 & 10 &   0.3 &  15.8 &   1.6 &  -7.5 &  29.5 &   1.7 &   4.8 &  21.0 &   0.3 &  -1.6 &  32.7 \\
		Hesse & 425 & 20 &   3.5 &   9.1 &   3.6 & -15.2 &  39.2 &   3.5 &   0.0 &  17.9 &   3.5 &   5.0 &  36.1 \\
		Mecklenburg-Vorpommern & 780 & 6 &   8.7 &   8.9 &   8.7 & -12.1 &  24.3 &   8.7 &  -6.4 &  15.3 &   8.7 &  -1.1 &  24.1 \\
		Lower Saxony & 1001 & 28 &   1.0 &   7.8 &   2.4 & -14.3 &  42.9 &   1.5 &  -3.1 &  18.5 &   1.1 &  -8.9 &  39.9 \\
		North Rhine-Westphalia & 391 & 48 &   0.3 &   7.7 &   3.8 & -10.8 &  40.9 &   4.2 &   1.1 &  17.4 &   3.4 &   1.6 &  34.0 \\
		Rhineland-Palatinate & 2304 & 15 &   0.5 &   9.8 &   0.6 & -14.0 &  36.5 &   1.7 &   0.7 &  22.9 &   0.5 &   1.6 &  30.3 \\
		Saarland & 52 & 4 &   3.9 &   4.8 &   3.9 &   5.3 &  49.0 &   3.9 &  -0.8 &  13.0 &   3.9 &  22.7 &  31.3 \\
		Saxony & 437 & 13 &   1.8 &   6.7 &   2.4 &   1.2 &  38.2 &   1.8 &   6.8 &  16.7 &   1.8 &   0.2 &  28.4 \\
		Saxony-Anhalt & 222 & 9 &   3.6 &  10.3 &   7.1 &   1.5 &  41.2 &   4.3 &   4.1 &  13.6 &   3.6 &  -4.5 &  29.4 \\
		Schleswig-Holstein & 1109 & 11 &   1.2 &  11.8 &   1.4 &  -9.0 &  36.9 &   1.6 &   5.5 &  19.4 &   1.2 &  -1.7 &  27.5 \\
		Thuringia & 850 & 9 &   1.6 &   8.0 &   1.6 & -15.1 &  51.2 &   5.1 &   1.0 &  14.4 &   1.6 &   0.1 &  35.5 \\
	    \bottomrule
	    \end{tabular}%
	\caption{Overview of the resulting key figures for the different
	       approaches.  
	       \newline $m$, $k$: Number of municipalities / districts, respectively;
	        \newline State Dev: Absolute value of the deviation of the average district size in a state from the federal average;
	       \newline \o Dev is the average relative deviation of district sizes 
	       from the federal average;
	       \newline $\Delta$MoI signifies the relative change of the
	       		moment of inertia (see \eqref{eq:momentofinertia}) compared to 2013;
	 			  \newline $\Delta$Pairs gives the proportion of changed voter-pair
	 			  assignments (as defined in \eqref{eq:formula_changedpairs}). }
	 			  \label{tab:res:overview}%
		\end{minipage}
  }  
    }
    
     \end{table*}

\end{appendices}

\end{document}